\documentclass[%
 reprint,
superscriptaddress,
nofootinbib,
nobibnotes,
 amsmath,amssymb,
 aps,floatfix,
]{revtex4-2}

\usepackage{amsthm}
\usepackage{graphicx,xcolor}% Include figure files
\usepackage{dcolumn}% Align table columns on decimal point
\usepackage{bm}% bold math
\usepackage[breaklinks=true,colorlinks,citecolor=blue,linkcolor=blue,urlcolor=blue]{hyperref}% add hypertext capabilities
\usepackage[mathlines]{lineno}% Enable numbering of text and display math
\usepackage[caption=false, position=bottom]{subfig}
\usepackage{floatrow}
\usepackage{mathtools}

\DeclareMathOperator{\Var}{Var}

\usepackage[utf8]{inputenc}
\usepackage[english]{babel}

\newtheorem{theorem}{Theorem}
\begin{document}

\title{Quantum Energy Landscape and VQA Optimization}

\author{Joonho Kim}
% \email{joonhokim@ias.edu}
\affiliation{
 School of Natural Sciences, Institute for Advanced Study, Princeton, NJ 08540, USA.
}
 
\author{Yaron Oz}%
\affiliation{
 School of Natural Sciences, Institute for Advanced Study, Princeton, NJ 08540, USA.
}
\affiliation{
 Raymond and Beverly Sackler School of Physics and Astronomy, Tel-Aviv University, Tel-Aviv 69978, Israel.
}

\begin{abstract}
We study the effects of entanglement and control parameters on the energy landscape and optimization performance of the variational quantum circuit. Through a systematic analysis of the Hessian spectrum, we characterize the local geometry of the energy landscape at a random point and along an optimization trajectory. We argue that decreasing the entangling capability and increasing the number of circuit parameters have the same qualitative effect on the Hessian eigenspectrum. Both the low-entangling capability and the abundance of control parameters increase the curvature of non-flat directions, contributing to the efficient search of area-law entangled ground states as to the optimization accuracy and the convergence speed.
% The structure of the Hessian spectrum of eigenvalues is correlated with 
% the entanglement of the corresponding quantum state, where
% low entanglement being characterized by a small number of large outlier eigenvalues and a large bulk of small eigenvalues.  
% We observe two phases during optimization where the first one is performed on a subspace spanned by the large eigenvalue vectors of the Hessian while  the second one takes place a subspace spanned by its small eigenvalue vectors.
% We find a generic growth of the Hessian 
% top eigenvalue along the optimization trajectory.
% We show that for a fixed entanglement a large number of parameters leads to better optimization performance, which is reminiscent of overparametrization in deep neural networks.
% We observe that fixing the circuit entanglement and increasing the number of parameters has
% the same qualitative effect on the Hessian eigenspectrum as fixing the number of parameters and decreasing
% the entanglement.

\end{abstract}

\maketitle

\section{Introduction}

% * Variational Optimization

The variational quantum algorithm (VQA) is arguably the most promising framework to achieve the near-term quantum advantage \cite{VQE2014,QAOA}. 
The structure of typical VQA computation consists of three parts: 
First, the quantum processor constructs the wavefunction $|\psi(\mathbf{\bm\theta})\rangle$ by acting a sequence of unitary gate operations, which often depend on randomly chosen control parameters $\bm{\theta}$, on the initial product state $|0\rangle^{\otimes n}$. 
Second, the quantum processor measures the variational wavefunction, where outputs of the repeated measurements, e.g., $\langle \psi(\bm\theta) | Z_i | \psi(\bm\theta)\rangle$ for $1 \leq i \leq n$,
are passed to the classical processor for quantum state tomography. 
Third, the classical processor estimates the energy function $\mathcal{L}(\bm\theta) \equiv \langle \psi(\bm\theta) | \mathcal{H} | \psi(\bm\theta) \rangle $, where $\mathcal{H}$ is the Hamiltonian that encodes a given problem,  and searches an optimal parameter $\bm\theta^* = \arg\min_{\bm\theta} \mathcal{L}(\bm\theta)$ that minimizes it. Such optimization is typically done by the local gradient search that requires the iterative evaluation of the energy function and the updated parameter. See \cite{Endo_2021,cerezo2020variational} for the recent reviews on the VQA algorithms.

At the heart of these VQA approaches lies the variational circuit ansatz that generates quantum wavefunctions  depending on a set of control parameters stored and manipulated in classical devices.  Common choices of unitary gates are usually limited to one-qubit rotation gates and two-qubit entangling gates acting only upon adjacent qubit pairs for the feasibility of hardware implementation \cite{Kandala_2017_hardware}.
There are numerous ways to design the variational circuits even within this limited class. % If symmetry properties of the desired quantum state are known in advance, it is sometimes possible to find a problem-tailored circuit ansatz \cite{Liu_2019_symmetry,Gard_2020_symmetry}. However, 
In most applications, we rely on the heuristic approach to find an effective circuit ansatz whose expected performance is not a priori known. The goal of this paper is to bring design principles for an efficient variational ansatz concerning its entangling capability and number of control parameters,
% , and the locality of the objective function 
by measuring how these factors influence the quantum energy landscape defined by $\mathcal{L}(\bm\theta)$ and the performance of parameter optimization.

% * Entanglement
The overwhelming majority of the Hilbert space is occupied by highly entangled generic quantum states that exhibit the volume-law scaling of entanglement entropies, i.e., proportional to the number of subsystem qubits. 
As a result, the broader range of states the circuit ansatz $|\psi(\bm\theta)\rangle$ can express, the higher the mean entanglement entropy over randomly sampled states $\{| \psi(\bm\theta_s)\rangle\}_s$ becomes. In this regard, the average entanglement entropy $\mathcal{R}^{(k)}$ of the circuit generated states can represent the expressibility of the variational ansatz. 

There is, however, a negative correlation between the average entanglement entropy and the optimization success of the circuit parameters $\bm\theta$ towards the ground states \cite{holmes2021connecting,Kim:2021ffs}. 
% It comes from the following reason: 
The area law scaling of the entanglement entropy is a commonly expected correlation pattern of the ground states of local gapped Hamiltonians \cite{Hastings_2007}. In contrast, most of the circuit parameter space is associated with highly-entangled typical quantum states, especially if the mean entanglement entropy of the variational circuit is close to the maximum. That makes the local parameter search of a highly expressible ansatz less likely to succeed in finding a trajectory towards the low-entangled ground states.

A closely related geometric statement is known as the barren plateau theorem:
Assuming the 2-design characteristic of the random circuit ensemble, the gradient of the energy function $\mathcal{L}(\bm\theta)$ with respect to the circuit variables $\bm\theta$ is zero on average, with the variance exponentially suppressed for the growing system size $n$ \cite{McClean2018bp}. This has been shown in \cite{arrasmith2021equivalence} to be equivalent to the exponential decay of $\text{Var}_{\bm\theta}\left[\mathcal{L}(\bm\theta + \bm\alpha) - \mathcal{L}(\bm\theta)\right]$ with respect to $n$, indicating how generically flat the quantum energy landscape of the highly expressible variational circuit is. In this work, we will further investigate how flatness of the quantum landscape is correlated with the entangling capability of the circuit, by examining the local geometry near generic random points as well as certain special points extracted from the parameter optimization trajectory.

% * Parameter Space
Since the quantum states generated by the circuit are controlled by continuous variables $\bm\theta$  stored and manipulated in the classical computer, the number of classical parameters can also be a crucial factor that affects the energy landscape and optimization performance. One extreme case was studied in \cite{highdepth}: The local gradient search for the over-parameterized circuit can approximate the Hamiltonian ground state very precisely, in both cases where the ground state entanglement entropy follows the volume-law scaling and the area-law scaling, despite the high expressibility of the variational circuit ansatz. More systematically, given a fixed amount of the average entanglement entropies, one can vary 
%the circuit structure and consequently change 
the number of control variables by adding single-qubit Pauli rotation gates.
%and/or removing two-qubit entanglers. 
We will quantify how it affects the flatness of the energy landscape and the optimization performance. 

Our study will be conducted by investigating the Hessian matrix of the energy function, $H_{ab} (\bm\theta) = {\partial_{a}\partial_{b}} \mathcal{L}(\bm\theta)$, evaluated at random initial points \cite{cerezo2020hessian_bp}, final convergence points \cite{highdepth}, and multiple intermediate points chosen from the optimization trajectory \cite{Huembeli_2021_hessian_convergence}. The Hessian eigenvalue spectrum reveals the information about the shape of the quantum energy landscape and how the local parameter search works. 
We will characterize an important relationship between its eigenvalues and the entangling capability of the circuit: low-entangling circuits gives a massive bulk of small Hessian eigenvalues and a few large outliers. 
Such spectral pattern is often observed in classical deep neural networks with over-parameterization  \cite{sagun2017eigenvalues,ghorbani2019investigation,fort2019emergent}. 
We will indeed observe the close similarity between the spectral evolution caused by adding more circuit parameters and by reducing the circuit entangling capability.
% where it was correlated with the network's over parametrization.
% Indeed, there the gradient descent takes place only within a small number of sub-dimensions \cite{gurari2018gradient}
% corresponding to the large eigenvalues.

%For instance, {\color{purple}the gradient descent takes place only within a small number of sub-dimensions %\cite{gurari2018gradient}. We also want to figure out whether there is an intrinsic dimension of the parameter space, for a given %Hamiltonian system, such that adding extra parameters beyond that upper bound does not improve the optimization performance and %remains a simple redundancy. (Revise later)}

The rest of the paper is organized as follows: Section~\ref{sec:ansatz} introduces the basic form of the variational circuit used in this paper and analyzes correlation functions of the circuit density matrices.
% We will prove a theorem that characterizes this correlation.
Section~\ref{sec:entanglement} systematically studies how stochastic dropout of two-qubit entangling gates affects the VQA performance of the variational circuit and the quantum energy landscape through the numerical evaluation of the Hessian matrix. It includes theorems on the top Hessian eigenvalue as well as the optimization rate. 
The impact of the control parameters on the VQA performance and the shape of the quantum energy landscape is studied by adding single-qubit rotation gates in Section~\ref{sec:param}. Finally, Section~\ref{sec:discussion} summarizes and provides suggestions for future research.

\section{Variational Circuit Ansatz}
\label{sec:ansatz}

\subsection{Circuit Architecture}

The parametrized quantum circuit used in this paper for various numerical experiments is a hardware-efficient ansatz \cite{Kandala_2017_hardware}, made of 1-qubit Pauli-Y rotations and 2-qubit CZ entanglers acting on  qubit pairs,
\begin{align}
\begin{split}
    R_{y,i}(\varphi) &= \left[e^{i \sigma_y \varphi}\right]_i,\\
    CZ_{i,j} &= \text{diag}(1, 1, 1, -1)_{i, j}.
\end{split}
\end{align}
The basic building block of our unitary circuit in its primary form is the 2-qubit unitary operator of Figure~\ref{fig:circuit_diagram_right}, 
\begin{align}
    U_{i,j}(\varphi_{a},\varphi_{b}) = CZ_{i,j} \cdot \left(R_{y,i}(\varphi_a)\otimes R_{y,j}(\varphi_{b})\right),
    \label{eq:gate}
\end{align}
acting on the 4-dimensional hyperplane spanned for the $(i,j)$ qubit pair, embedded in the $n$-qubit Hilbert space. 

The operators \eqref{eq:gate} acting on consecutive $(i, i+1)$ qubits compose together the following layer unitary operators:
\begin{align}
    U_\ell =
    \begin{cases}
    \bigotimes_{m=1}^{\lfloor n/2 \rfloor} U_{2m-1, 2m}(\varphi_{\ell, 2m-1}, \varphi_{\ell, 2m}) & \text{odd $\ell$}\\
    \textstyle\bigotimes_{m=1}^{\lfloor n/2 \rfloor} U_{2m, 2m+1}(\varphi_{\ell, 2m}, \varphi_{\ell, 2m+1})&\text{even $\ell$}  
    \end{cases}
\end{align}
where the periodic boundary condition $i \simeq i +n$ is imposed on the $n$-qubit lattice. It is  convenient to use the collective notation $U_\ell(\bm\varphi_\ell)$ where ${\bm\varphi_\ell}$ denotes all $\{{\varphi_{\ell,i}}\}_{i=1}^n$. The variational circuit states (see Figure~\ref{fig:circuit_diagram_left} for an illustration of the circuit architecture) are then generated by sequentially acting $L$ instances of the layer unitary operators on the initial product state $|0\rangle^{\otimes n}$, i.e., 
\begin{align}
    |\psi(\bm\theta)\rangle= U_L({\bm\varphi}_L) \cdots U_1({\bm\varphi}_1) |0\rangle^{\otimes n} = U({\bm\theta})  |0\rangle^{\otimes n} 
    \label{eq:circuit_state},
\end{align}
where $\bm\theta = \{{\bm\varphi}_\ell\}_{\ell=1}^L$. We will index the $nL$ components of the circuit parameter $\bm\theta$ by $1\leq a,b,\cdots \leq nL$.

% Our random quantum circuit in Figure~\ref{fig:circuit_diagram} is characterized by the number of qubits $n$, the number
% of random angle parameters $N$ and the number of Control-Z gates $m$.
% In our architecture we can control the entanglement entropy by the number of CZ gates that
% we apply, i.e. at each odd (even) time step a CZ gate is
% applied to qubits $(i,i+1),~ $i$~odd~ (even)$ with probability $p$. 
% The number of Control-Z gates is a random variable  we will refer by $m$ to its mean value.
% The size of the parameter space is controlled by the number of single qubit rotation gates that 
% we apply.
% The number of circuit layers  is denoted by $L$ and in the standard $p=1$ architecture reads: $L= \frac{N}{n}$.  
% The state of the circuit is $|\psi_c\rangle$ is a function of these parameters.
% We showed that the entanglement spectrum can serve as a diagnostic
% for the optimization performance. The entanglement spectrum is a function of $(n,m,N)$.
% We would like to perform a systematic study of the quantum circuit performance as a function of these parameters.

\begin{figure}[t]
\centering
\subfloat[Architecture]{
        \centering
        \label{fig:circuit_diagram_left}
        \includegraphics[height=1.85cm]{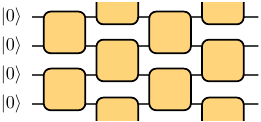}}
\hfill
\subfloat[Gate]{
        \centering
        \label{fig:circuit_diagram_right}
        \includegraphics[height=1.85cm]{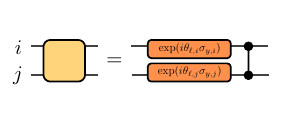}}    
\caption{The variational quantum circuit used in this paper.}
\label{fig:circuit_diagram}
\end{figure}
 
%In the limit $p\rightarrow 0, n\rightarrow \infty, np = 2\lambda = fixed$ we get at each time step a Poisson %distribution with mean $\lambda$.

\subsection{Circuit Density Matrix}
The variational circuit generates a pure quantum state \eqref{eq:circuit_state} whose corresponding density matrix is given by $\rho_{\alpha\beta}(\bm\theta) = U_{\alpha 1}(\bm\theta)U^{*}_{\beta 1}(\bm\theta)$ for all $1 \leq \alpha,\beta \leq 2^n$. In particular, the circuit unitary matrix  $U_{\alpha\beta}(\bm\theta)$ of Figure~\ref{fig:circuit_diagram} is real-valued and orthogonal, being parameterized by $nL$ circular variables $\{\theta_a\}_{a=1}^{nL}$ that have the period of $\pi$ \cite{fontana2020optimizing}. 
The associated parameter space is therefore the compact torus $T^{nL}$. We find that the density matrix $\rho_{\alpha\beta}(\bm\theta)$ can be written as follows in its Fourier expansion form:
 \begin{equation}
 \rho_{\alpha\beta}= \frac{\delta_{\alpha\beta}}{2^n} +  \sum_{q}  c^{q}_{\alpha\beta} \prod_{a=1}^{nL} \left(\sin{(2\theta_a)}\right)^{q_{2a-1}} \left(\cos{(2\theta_a)}\right)^{q_{2a}}
 \label{L}
 \end{equation}
 where the sum is taken over the set of $(2nL)$-dimensional discrete vectors, $q\in \{0,1\}^{2nL}$, except the zero $\{0^{2nL}\}$.
 
The expectation value of the density matrix $\rho_{\alpha\beta}(\bm\theta)$ with respect to the uniform measure on $\bm\theta \in T^{nL}$ reads:
\begin{equation}
\mathbb{E}_{\bm\theta}[\rho_{\alpha\beta}(\bm\theta)] = {\delta_{\alpha\beta}}/{2^n} \ ,
\label{er}
\end{equation}
where we used the orthogonality of the sine and cosine functions. Consequently, the expectation values of the energy function $\mathcal{L}(\bm\theta)$ and its derivatives are given by:
\begin{align}
    \label{eq:bp-energy-mean}
    \mathbb{E}_{\bm\theta} [\mathcal{L}(\bm\theta)]  &= \text{Tr}(\mathcal{H}) / 2^n \ ,  \\
    \mathbb{E}_{\bm\theta} [\partial_a\mathcal{L}(\bm\theta)]   &= 
    \mathbb{E}_{\bm\theta} [\partial_a\partial_b\mathcal{L}(\bm\theta)]  =  \cdots = 0 \ .
    \label{eq:bp-mean}
\end{align}
 
Computation of the second-order correlation functions of the energy function $\mathcal{L}(\bm\theta)$ and its derivatives requires the knowledge of two-point functions, $\mathbb{E}_{\bm\theta}[\rho_{\alpha\beta}(\bm\theta)\rho_{\rho\sigma}(\bm\theta)]$, of the density matrices.

\begin{theorem}
The two-point correlation function of the density matrix $\rho_{\alpha\beta}(\bm\theta)$ satisfies:
\begin{eqnarray}
  \mathbb{E}_{\bm\theta}[\rho_{\alpha\beta}(\bm\theta)\rho_{\rho\sigma}(\bm\theta)] &=& A_{\alpha\beta}(\delta_{\alpha\rho}\delta_{\beta\sigma} + \delta_{\alpha\sigma}\delta_{\beta\rho}) + \nonumber \\  &+&  A_{\alpha\rho} \delta_{\alpha\beta}\delta_{\rho\sigma} \ .
 \label{rho2pt}
\end{eqnarray}
There is no summation over repeated indices in (\ref{rho2pt}).
\end{theorem}

\begin{proof}
For real symmetric matrices $\rho_{\alpha\beta}$, the general form of the two-point functions reads:
\begin{equation}
 \mathbb{E}_{\bm\theta}[\rho_{\alpha\beta}(\bm\theta)\rho_{\rho\sigma}(\bm\theta)] = A_{\alpha\beta}  (\delta_{\alpha\rho}\delta_{\beta\sigma} +\delta_{\alpha\sigma}\delta_{\beta\rho})  + B_{\alpha\rho} \delta_{\alpha\beta}\delta_{\rho\sigma}.
 \label{rho}
\end{equation}
It should satisfy the following relations:
\begin{align}
\begin{split}
\mathbb{E}_{\bm\theta}[\rho_{\alpha\alpha}(\bm\theta)\rho_{\beta\beta}(\bm\theta)] &= \mathbb{E}_{\bm\theta}[\rho_{\alpha\beta}(\bm\theta)\rho_{\beta\alpha}(\bm\theta)]\\
&= \mathbb{E}_{\bm\theta}[\rho_{\alpha\beta}(\bm\theta)\rho_{\alpha\beta}(\bm\theta)],
\label{cond}\end{split}\\
\textstyle\sum_{\beta=1}^{2^n} \mathbb{E}_{\bm\theta}[\rho_{\alpha\beta}(\bm\theta)\rho_{\beta\sigma}(\bm\theta)] &=  \mathbb{E}_{\bm\theta}[\rho_{\alpha\sigma}(\bm\theta)]\label{cond0}
\end{align}
% which is a consequence of the real-valuedness, $\rho_{ij} = \rho_{ij}^*$, of the circuit ansatz \eqref{eq:circuit_state} used throughout this paper. 
where the second equality \eqref{cond0} reflects the purity of the density matrix. By substituting \eqref{rho} into \eqref{cond}, we obtain the relation $A_{\alpha\beta} = B_{\alpha\beta}$ for $\alpha\neq\beta$. There is a redundancy in keeping both $A_{\alpha\alpha}$ and $B_{\alpha\alpha}$, because only the combination $2A_{\alpha\alpha}+B_{\alpha\alpha}$ appears independently for $\alpha=\beta$. Without loss of generality, we can rewrite \eqref{rho} as (\ref{rho2pt})
%\begin{equation}
%  \mathbb{E}_{\bm\theta}[\rho_{\alpha\beta}(\bm\theta)\rho_{\rho\sigma}(\bm\theta)] = A_{\alpha\beta} %(\delta_{\alpha\rho}\delta_{\beta\sigma} +\delta_{\alpha\sigma}\delta_{\beta\rho}) + A_{\alpha\rho} \delta_{\alpha\beta}\delta_{\rho\sigma}.
% \label{rho2}
%\end{equation}
for which \eqref{cond0} becomes equivalent to 
\begin{align}
    \textstyle(3A_{\alpha\alpha} + \sum_{\beta \neq \alpha} A_{\alpha\beta})\delta_{\alpha \sigma} = 2^{-n}\,\delta_{\alpha\sigma} \ .
    \label{eq:cond2}
\end{align}

\end{proof}

One can analytically derive the coefficients $A_{\alpha\beta}$ in two limiting cases. When two-qubit  entanglers are completely omitted from the circuit ansatz \eqref{eq:circuit_state}, the coefficients are
\begin{align}
    A_{\alpha\beta} = 
    \begin{dcases}
    3^{c_0(\alpha_2 \vee \beta_2) + c_1(\alpha_2 \wedge \beta_2)}/{2^{3n}} & \text{for } \alpha \neq \beta\\
    3^{n-1} / 2^{3n} & \text{for } \alpha =\beta \ ,
    \end{dcases}
    \label{eq:coeff-product}
\end{align}
where the subscripts $2$ in $\alpha_2$ and $\beta_2$ indicate that $\alpha$ and $\beta$ are in their binary representation. The function $c_{0/1}(s_2)$ gives the number of $0/1$'s in the input binary string $s_2$. We also have checked that \eqref{eq:coeff-product} satisfies \eqref{eq:cond2} by manipulating symbolic expressions in \texttt{Mathematica}.

Another case that allows the exact analysis is when the circuit distribution is the Haar orthogonal ensemble, for which one can replace the integration over $T^{nL}$ with the Haar integral over
the orthogonal group $O(2^n)$. The one-point and two-point functions are obtained by the Weingarten calculus \cite{Collins_2006} as:
\begin{align}
    \label{eq:bp-ortho0}
    \mathbb{E}_{\rho \in O(2^n)}[\rho_{\alpha\beta}] &= 2^{-n} \delta_{\alpha\beta} \ , \\
    \mathbb{E}_{\rho \in O(2^n)}[\rho_{\alpha\beta}\rho_{\rho\sigma}] &= \frac{\delta_{\alpha\beta}\delta_{\rho\sigma} + \delta_{\alpha\rho}\delta_{\beta\sigma} + \delta_{\alpha \sigma}\delta_{\beta \rho}}{2^n(2^n+2)} \ .
    \label{eq:bp-ortho}
\end{align}
Note that \eqref{eq:bp-ortho} implies the existence of the barren plateau problem \cite{McClean2018bp} for the orthogonal 2-design ensemble. When the variational circuit  \eqref{eq:circuit_state} behaves as an approximate orthogonal 2-design, the variance of random energy gradients decays exponentially with the system size $n$ as:
\begin{equation}
   \text{Var}_{\rho \in O(2^n)}[\partial_a {\cal L}]  \sim \frac{\text{Tr}(\mathcal{H}^2)}{4^n} \sim \mathcal{O}(2^{-n}) \ ,
   \label{eq:bp-vargrad}
\end{equation}
assuming the scaling behavior $\text{Tr}(\mathcal{H}^2) \sim \mathcal{O}(2^n)$ of various 1d spin-chain Hamiltonian systems.

\begin{figure}[t]
  \centering
  \subfloat[]{\includegraphics[height=4.1cm]{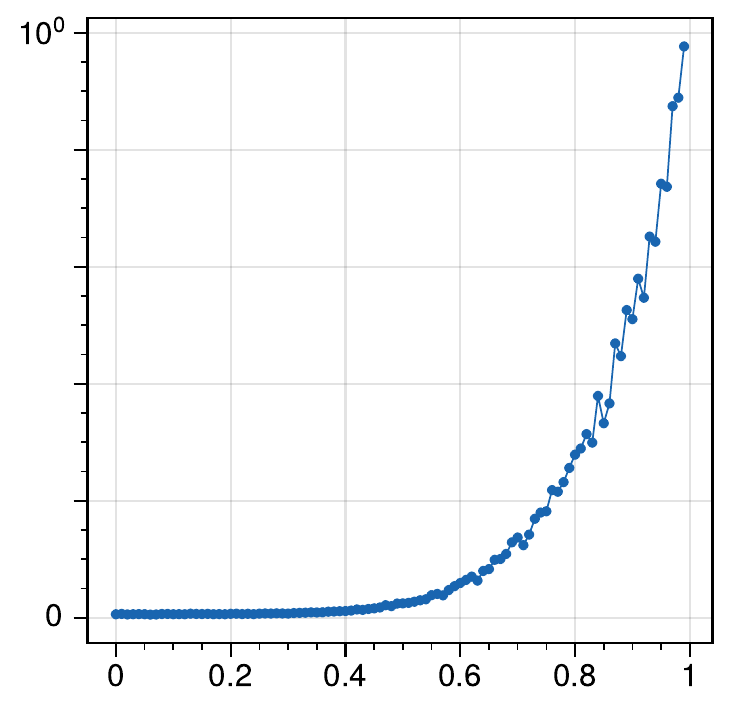}\label{fig:bp-p}}
  \subfloat[$\mathbb{E}_{\bm\theta}${[$\rho_{\alpha\beta}(\bm\theta)\rho_{\alpha\beta}(\bm\theta)$]}]{\includegraphics[height=4.1cm]{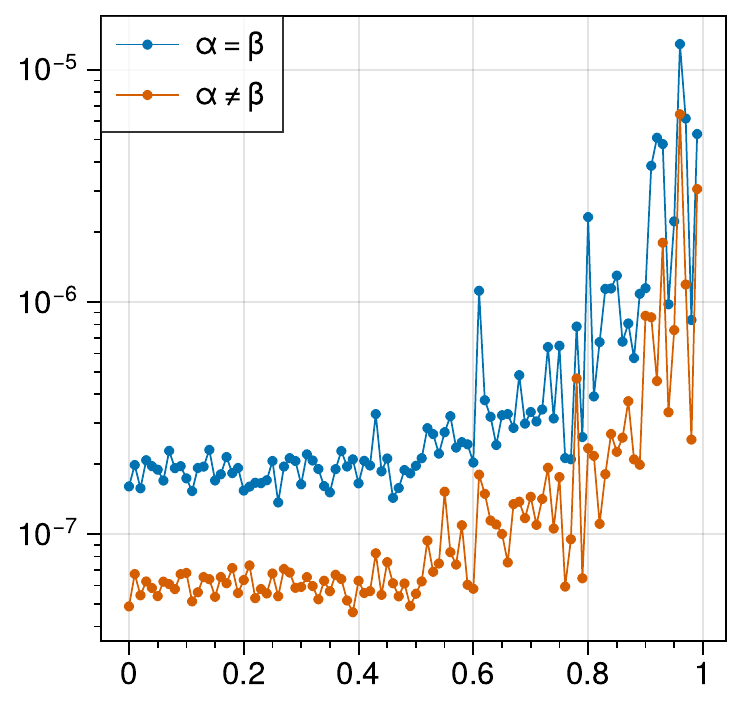}\label{fig:bp-var}}
  \caption{Correlation functions of the circuit density matrix at $L=56$, estimated from $1500$ samples for each probability $p$ that removes the two-qubit CZ entanglers.}
  \label{fig:bp-entropy}
\end{figure}

In Section~\ref{sec:entanglement}, we will explore systematic reduction of the average entanglement entropy in the circuit states \eqref{eq:circuit_state} by randomly and repeatedly removing the CZ entanglers with the probability $p$. Given sufficient circuit depth, the case with $p=0$ corresponds to \eqref{eq:bp-ortho} that follows the Haar orthogonal ensemble, while the case with $p=1$ leads to the $n$-qubit product state for which we find  \eqref{eq:coeff-product}. We can  interpolate these two extreme cases through numerical estimation of $\text{Var}_{
\bm\theta}[\partial_a {\cal L}]$ and $\mathbb{E}_{\bm\theta}[\rho_{\alpha\beta}(\bm\theta)\rho_{\alpha\beta}(\bm\theta)]$ for $0 \leq p < 1$. Specifically in the $L=56$ case, the results are summarized in Figure~\ref{fig:bp-entropy}.

\section{Entanglement, Energy Landscape and Optimization}
\label{sec:entanglement}

\begin{figure*}[t]
\centering
\subfloat[Energy difference $\Delta E$]{
    \includegraphics[height=4.1cm]{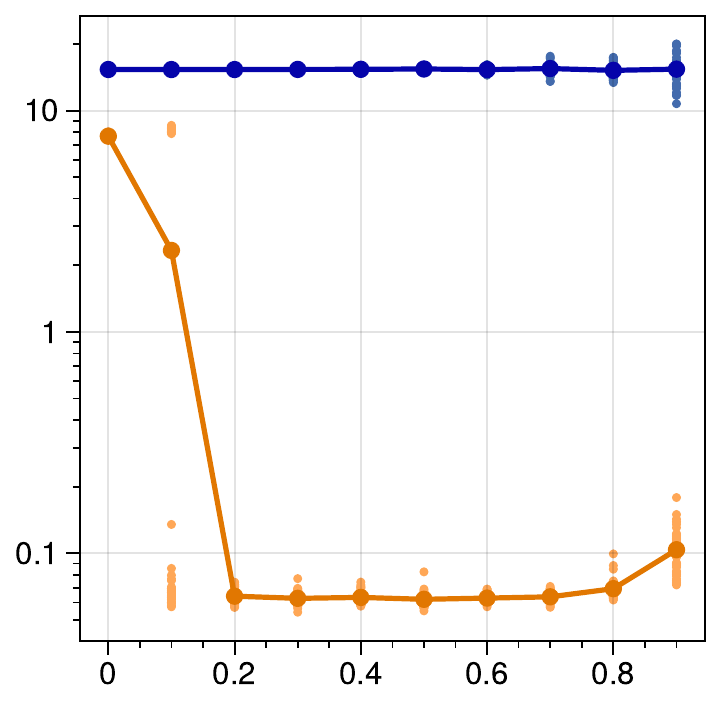}\label{fig:optim_ising_g1_left}}
\subfloat[Renyi-2 entropy $\mathcal{R}^{(2)}$]{
    \includegraphics[height=4.1cm]{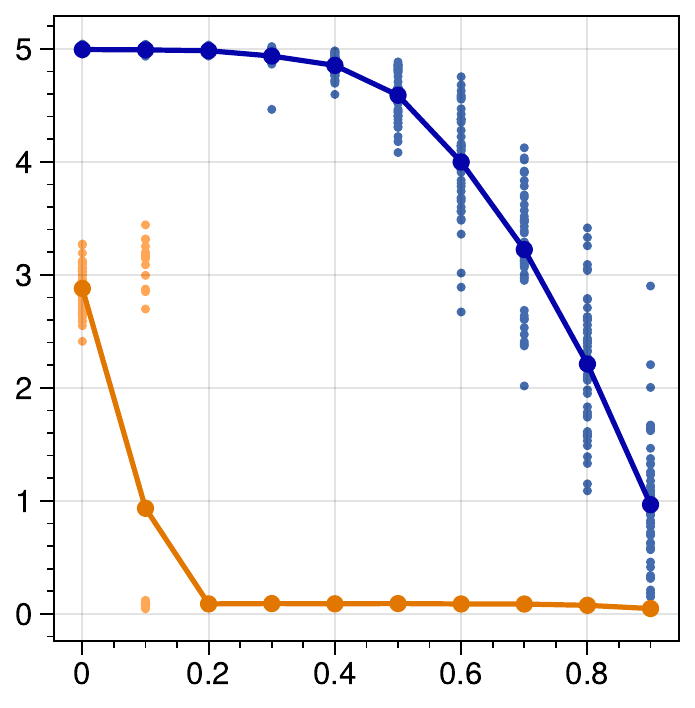}\label{fig:optim_ising_g1_middle}}
\subfloat[\% that reaches $\Delta E < 0.1$]{
    \includegraphics[height=4.1cm]{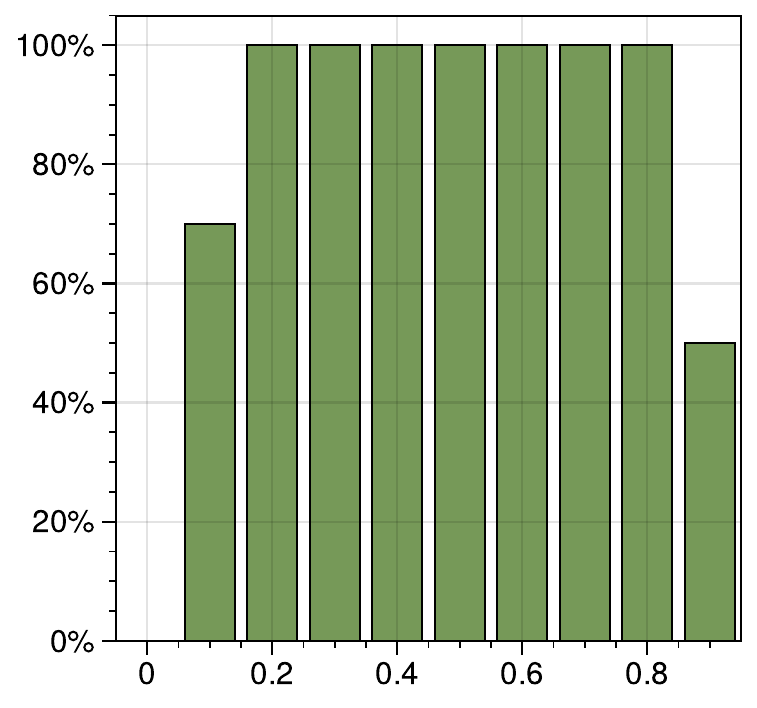}\label{fig:optim_ising_g1_success}}
\subfloat[Steps $\tau$ until $\Delta E < 0.1$]{
    \includegraphics[height=4.1cm]{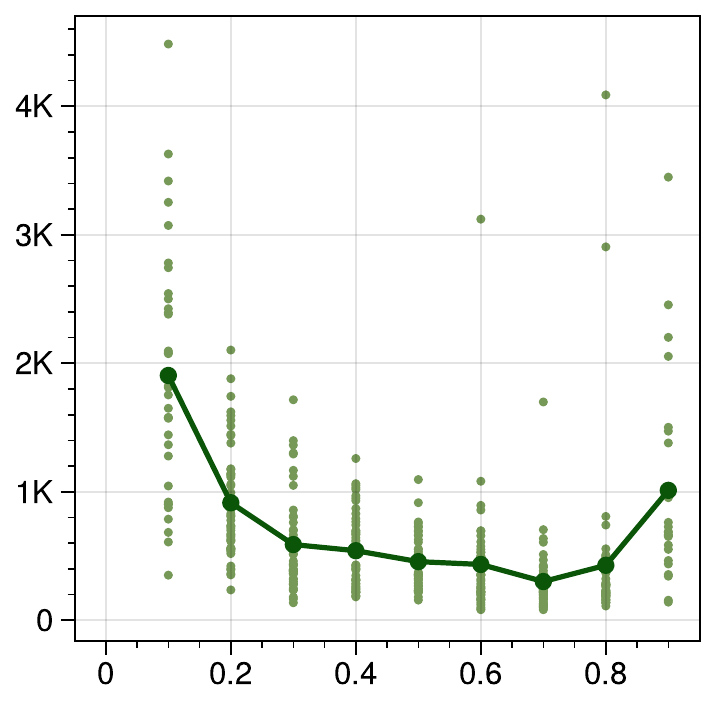}\label{fig:optim_ising_g1_right}}
\caption{Collections of $50$ VQA instances to approximate the $n=12$ Ising ground state with the circuits with $L=56$ layers, for each probability $p$ of omitting the CZ gates. The results are successful for $p \in [0.2,0.8]$.}
%  \caption{The VQA result with $n=12$ qubits and $L=56$ layers as a function of the probability $p$ to omit the a CZ 2-qubit gate.  We see a successful optimization for $p \in [0.2,0.8]$. (a) The energy difference between the optimized circuit and the ground state.
%     (b) The Renyi-2 entropy of the optimized circuit. (c) The percentage of successful optimizations. (d) The required number of optimization steps.}
    \label{fig:optim_ising_g1}
\end{figure*}

This section will explore how the shape of the quantum energy landscape varies with different levels of entangling capability. We will consider the average geometry at random generic parameters \cite{McClean2018bp,cerezo2020hessian_bp} as well as the local geometry around optimization trajectories \cite{Huembeli_2021_hessian_convergence}.
We will fix the number of the circuit control parameters
%, whose effect on the energy landscape will be studied in Section~\ref{sec:param}, 
and systematically vary the entangling capability of the circuit. That can be done by dropping out the CZ entanglers contained in the two-qubit operator \eqref{eq:gate} with probability $p$. 

%\subsection{Circuit Entanglement? \color{red}NAME}
\subsection{Circuit State Entanglement} 

We define the entangling capability of the variational circuit as the average entanglement entropy over the circuit state ensemble, $\{|\psi(\bm\theta)\rangle : \bm\theta \in [0, 2\pi)^{\otimes nL} \}$, estimated through the sample average over $M$ circuit states:
\begin{align}
    %\int \frac{d\theta}{(2\pi)^d} \, \mathcal{R}^{(k)} (|\psi(\theta)\rangle) \simeq
    \frac{1}{M}\sum_{q=1}^M \mathcal{R}^{(k)} (|\psi(\bm\theta_q)\rangle) \ \text{ where }\ \bm\theta_q \sim \mathcal{U}(0,2\pi)^{\otimes nL} \ .
\end{align}
As typical quantum states that comprise an exceedingly large portion of the Hilbert space are highly-entangled, it represents how expressible the variational ansatz is, i.e., how various quantum states $|\Phi\rangle$ can be closely approximated by the circuit ansatz within tolerance $\varepsilon$, 
\begin{align}
    \|{|\Phi\rangle - |\psi(\bm\theta^*)\rangle}\| < \varepsilon,
\end{align}
at a certain parameter $\bm\theta^*  \in [0, 2\pi)^{\otimes nL}$.

It was shown in \cite{Kim:2021ffs} that the average entanglement entropy of the dense circuit at $p=0$ grows \emph{linearly} for an increasing circuit depth $L$, then saturates to a maximum possible value $n_A - c_k$ beyond a critical depth $L_s < L$. $n_A$ and $c_k$ denote the subsystem size and a non-negative constant that depends on the circuit architecture and the order $k$ of the entanglement entropy $\mathcal{R}^{(k)}$, respectively. Since the saturation depth $L_s$ itself also scales linearly with the system size $n$ \cite{Kim:2021ffs}, the number of two-qubit CZ gates introduced until the saturation of the entangling capability is roughly $m_s=O(n^2)$. We will choose the circuit depth $L^*$ such that the mean number of the CZ gates in the stochastic circuit $m =\tfrac{1}{2} n L^* (1-p)$ can be 
\begin{align}
    \begin{cases}
    m \gtrsim m_s & \text{for } p \rightarrow 0\\
    m < m_s & \text{for } p \rightarrow 1 \ . 
    \end{cases}
\end{align}
Specifically, $L^* =56$ will be sufficient for our purposes.

\subsection{Optimization Accuracy and Speed}

To reveal the connection between the entangling capability and VQA performance of the variational circuit, we consider solving the ground state of the most prototypical system, i.e., the 1d transverse-field Ising model:
\begin{align}
    \mathcal{H} = J \sum_{\langle i,j\rangle }Z_i Z_j +g \sum_i X_i \ \  \text{ with }\ \  J=1, \, g=1,
    \label{eq:ising}
\end{align}
and measure the deviation of the circuit energy from the exact ground-level energy,
\begin{align}
    \Delta E \equiv \langle \psi(\theta) | \mathcal{H} | \psi(\theta) \rangle - E_g \ .
\end{align}
Figure~\ref{fig:optim_ising_g1} is the collection of $50$ independent optimization results for the $L=56$ circuits with $p \in \{0, 0.1, \cdots, 0.9\}$. 

The curves in Figures~\ref{fig:optim_ising_g1_left}~and~\ref{fig:optim_ising_g1_middle} are respectively the energy gap $\Delta E$ and Renyi-2 entanglement entropy $\mathcal{R}^{(2)}$ evaluated under an equal partitioning of $n=12$ qubits. The blue/orange colors indicate whether the displayed values are before/after applying the gradient descent \eqref{gdm} to circuit parameters $\tau=5000$ times.
When the average entanglement entropy of the pre-optimization states saturates to the maximum possible value, as the cases for $p \leq 0.1$, Figure~\ref{fig:optim_ising_g1_left} exhibits the formulation of orange dot clusters around $\Delta E \sim 9$. It means the failure of many circuit instances in reducing $\Delta E$ via the local gradient search \eqref{gdm}. The gradient descent fails to make a trajectory towards the Ising ground state, while stopping at a suboptimal extremum in the quantum energy landscape. It happens because for the circuit with maximum entangling capability, finding the desired parameter $\bm\theta^*$ such that $|\psi(\bm\theta^*)\rangle \simeq |\Psi\rangle$ is roughly as hard as searching the Ising ground state $|\Psi\rangle$ over the entire Hilbert space. In contrast, the circuit with less entangling capability limits the local search to a subregion of the Hilbert space consisting of low-entangled states, which may still include the ground state, thereby facilitating its discovery \cite{Eisert_2010, Kim:2021ffs}.

Figures~\ref{fig:optim_ising_g1_success}~and~\ref{fig:optim_ising_g1_right} display two complementary metrics about the circuit performance with respect to the VQA optimization, representing how difficult the local gradient descent is to find the circuit parameter $\bm\theta^*$. When the entangling capability of the circuit is too low or too high, the VQA optimization may fail to approximate the ground state, and $\Delta E$ does not fall within a tolerance range.
Figure~\ref{fig:optim_ising_g1_success} shows the sample success rate of VQA trials lying within an acceptable error margin $\Delta E < 0.1$. It not only confirms the failure of maximally entangling circuits in finding the ground state, but also displays the dropping VQA performance of the low-entangling circuits with $p> 0.8$. Also consistent is the minimum number of parameter updates for each successful circuit instance to 
satisfy $\Delta E < 0.1$, as summarized in Figure~\ref{fig:optim_ising_g1_right}. Its mean and variance are minimized at $p=0.7$, while being larger as $p$ approaches the boundary value of $0.1$ or $0.9$. These results highlight that the VQA optimization works most efficiently with variational circuits in an intermediate range of the entangling capability \cite{Kim:2021ffs}.

A specific value of the optimal $p$ may vary for a different choice of the depth $L$ or target state $|\Psi \rangle$ that follows the area-law entanglement. Nevertheless, the basic shape of the curves will remain the same, and the circuit with medium expressibility will be most outperforming. 

% It is useful to note that the loss function is composed of a trace of the circuit density matrix with the Ising Hamiltonian
% , which 
% is effectively a sum of partial traces that dependent on the number of circuit layers.
% The dependence of the variance of the loss function gradient on the number of layers reveals the circuit state level of entanglement.

\begin{figure}[t]
    \centering
    \subfloat[$\tau=0$]{
    \includegraphics[height=4.1cm]{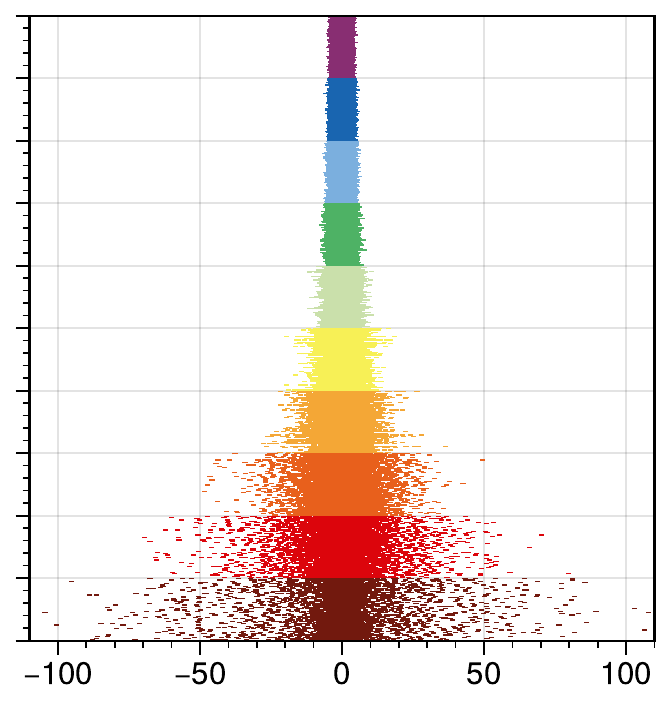}\label{fig:hess_ising_g1_init_spectrum}}
    \subfloat[$\tau=5\times 10^3$]{
    \includegraphics[height=4.1cm]{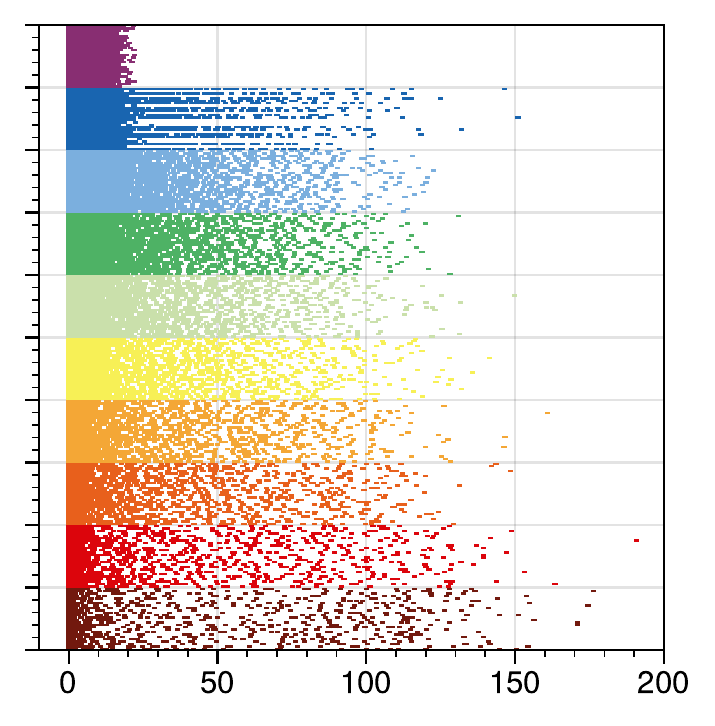}\label{fig:hess_ising_g1_final_spectrum}}\\
    \includegraphics[width=0.95\textwidth]{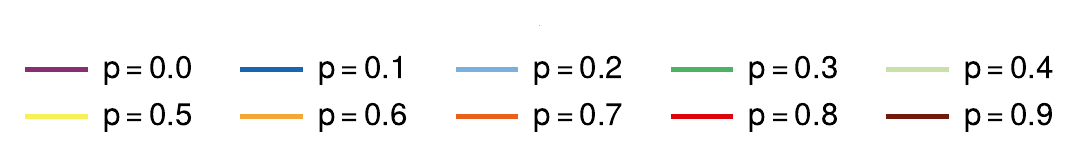}
    \caption{A visual collection of 500 sample Hessian eigenspectra at $L=56$ for different probabilities $p$ of omitting a CZ-gate, after taking $\tau$ parameter update steps.}
    \label{fig:hess_ising_g1_spectrum}
\end{figure}

%\subsection{Quantum Energy Landscape \color{red}NAME}
\subsection{Hessian Eigenspectrum and Landscape Geometry}

We recall from \eqref{eq:bp-ortho} that  high entangling capability incurs the barren plateau phenomenon \cite{McClean2018bp,cerezo2020hessian_bp,entanglement-bp} that any partial derivative of the energy function $\mathcal{L}(\bm\theta)$ becomes statistically zero \eqref{eq:bp-mean} on average over $\bm\theta \in T^{nL}$ with an exponentially decaying variance \eqref{eq:bp-vargrad} as to the system size $n$.
Since an initial gradient at arbitrarily chosen $\bm\theta$ vanishes with exponentially large probability, the gradient-based optimization typically cannot even start moving towards local minima in the large $n$ limit. 

However, the systematic control of $p$ can make the circuit ansatz move from the high entangling limit and dramatically improve the VQA performance, as evident in Figure~\ref{fig:optim_ising_g1}. Thus we want to further characterize the geometric implication of entangling capability, i.e., how exactly it eases the barren plateau problem and makes a large impact on the VQA performance. We will compare side-by-side the Hessian eigenspectra of the circuits with a fixed number of parameters and different values of $p$.

The repeated application of the parameter-shift rule allows us to express a higher-order derivative of the energy function as a combination of the first-order gradients. As a result, \cite{cerezo2020hessian_bp} found a probabilistic inequality for the second-order derivatives, which in turn leads to the following probabilistic bound on the Hessian eigenvalues.

\begin{theorem}
The largest absolute eigenvalue of the Hessian, $H_{ab} = \partial_{a}\partial_{b}\mathcal{L}(\bm\theta)$, is probabilistically bounded as
\begin{align}
 \Pr(|h_{\rm max}| \geq c ) \leq   \frac{2n^2 L^2 \Var_{\bm\theta}(\partial_a {\cal L}(\bm\theta))}{c^2} \ .
 \label{th1}
\end{align}
\end{theorem}
\begin{proof}
Let us denote the Hessian eigenvalues by $h_a$ where $1 \leq a\leq nL$. Since $\text{Tr}(H^2) = \sum_{a,b} H_{ab}^2 = \sum_{a} h_a^2$, 
\begin{eqnarray}
 |h_\text{max}| &\leq&  \textstyle  \big(\sum_a h_a^2\big)^{1/2} = 
 \big(\sum_{a,b} H_{ab}^2\big)^{1/2} \nonumber\\ &\leq&  nL\,  \max_{a,b} (|H_{ab}|) \ ,
 \label{1}
\end{eqnarray}
where $\max_{a,b} (|H_{ab}|)$ is the largest absolute value of the Hessian matrix elements.
Using the inequality that holds for every $1\leq a,b \leq nL$ \cite{cerezo2020hessian_bp},
\begin{equation}
    \Pr(|H_{ab}| \geq c) \leq \frac{2 \Var_{\bm\theta}(\partial_a {\cal L}(\bm\theta))}{c^2} \ ,
    \label{r2}
\end{equation}
we arrive at the probabilistic bound  (\ref{th1}) from (\ref{1}).
\end{proof}

\begin{figure*}[t]
\centering
\subfloat[Top/bottom eigenvalues]{
    \includegraphics[height=4.1cm]{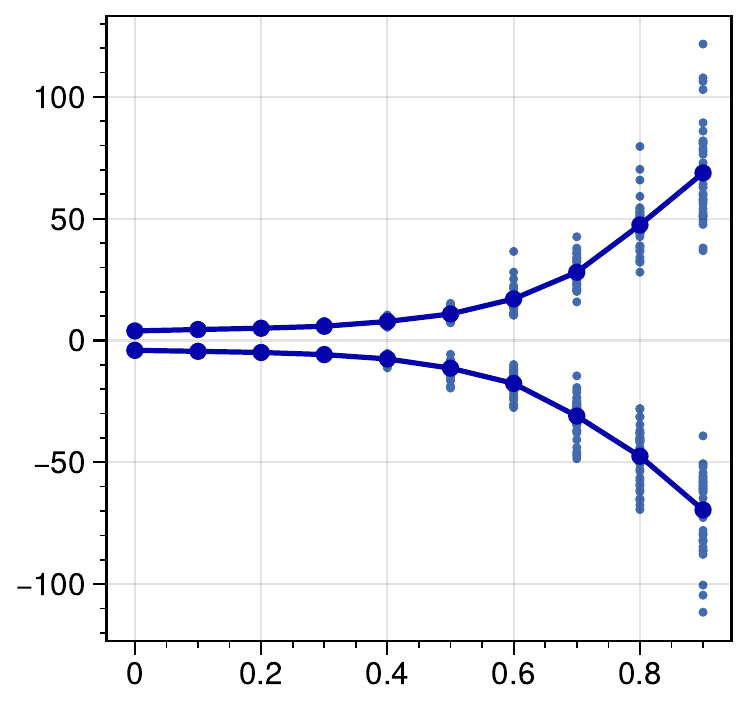}\label{fig:hess_ising_g1_init_outlier}}
\subfloat[\% of large eigenvalues ($|\lambda| > 5$)]{
    \includegraphics[height=4.1cm]{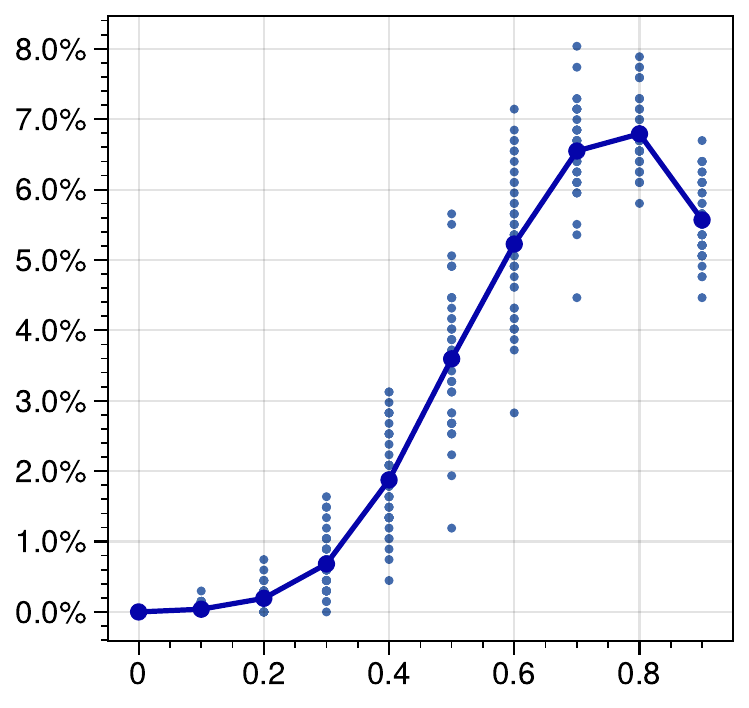}\label{fig:hess_ising_g1_init_large}}
\subfloat[\% of small eigenvalues ($|\lambda| < 0.2$)]{
    \includegraphics[height=4.1cm]{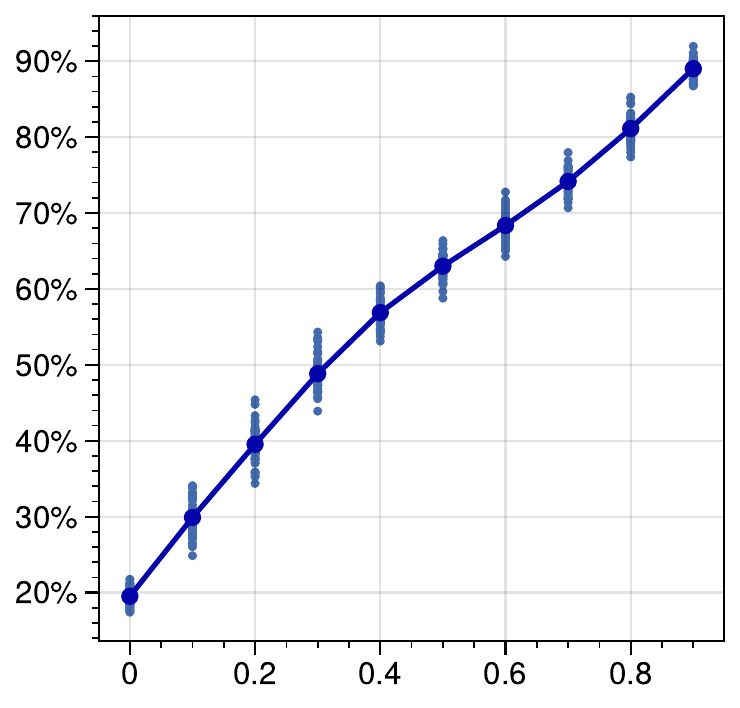}\label{fig:hess_ising_g1_init_small}}
\subfloat[Gradient overlap with $\mathcal{P}_{\text{small}}$]{
    \includegraphics[height=4.1cm]{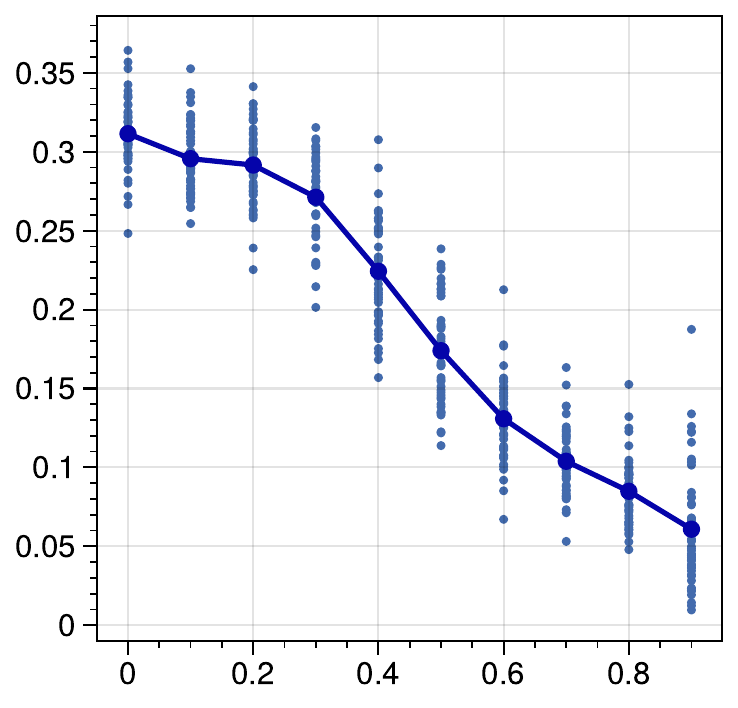}\label{fig:hess_ising_g1_init_overlap}}
    \caption{Characteristic plots for the Hessian eigenspectrum with $n=12$ qubits and $L=56$ layers, based on 50 instances for each probability $p$ to omit the CZ-gates, at randomly initialized circuit parameters.}
    % \caption{The Hessian eigenspectrum with $n=12$ qubits and $L=56$ layers as a function of the probability $p$ for omitting a CZ-gate (before optimization).}
    \label{fig:hess_ising_g1_init}
\end{figure*}

% We recall that in the high entangling regime of the circuit, i.e., $p \rightarrow 0$ and $L \gg 1$, where the circuit ensemble is approximately orthogonal 2-design, the variance of the energy gradient  decreases exponentially with respect to $n$. 
Figure~\ref{fig:hess_ising_g1_init_spectrum} visualizes a collection of Hessian eigenspectra evaluated at $50$ independent random circuit parameters for each $p$. Its horizontal axis denotes the eigenvalues, and the vertical axis extends across $500$ sample Hessians distinguished by colors according to their $p$ values. Their top/bottom eigenvalues are also depicted in  Figure~\ref{fig:hess_ising_g1_init_outlier} as a function of $p$. 
The consequence of the inequality \eqref{th1} is apparent in both  Figures~\ref{fig:hess_ising_g1_init_spectrum}~and~\ref{fig:hess_ising_g1_init_outlier}: While the largest absolute eigenvalues rapidly grow as $p \rightarrow 1$, all absolute eigenvalues at $p=0$ are bounded above as $|h_i| < 5$.
Such robust concentration of the Hessian spectrum towards $0$ by moving $p$ closer to $0$ shows that enhanced entangling capability causes a geometric crossover that smooths all the steepest directions in the quantum energy landscape.

% Interestingly, although the upper bound on the absolute Hessian eigenvalues becomes tighter as $p\rightarrow 0$, it does not indicate that more Hessian eigenvalues degenerate to $0$. 
Notably, the tighter concentration of Hessian eigenvalues as $p\rightarrow 0$ does not indicate higher degeneracy at zero eigenvalues. Let us count the number of stringent flat directions whose corresponding Hessian eigenvalues satisfy $|h_a| < 0.2$ and denote by $\mathcal{P}_\text{small}$ the flat subspace spanned by them. Figure~\ref{fig:hess_ising_g1_init_small} summarizes, for each different $p$, the average percentage between the dimensionality of $\mathcal{P}_\text{small}$ and that of the entire parameter space. Initially at $p=0$, the flat subspace $\mathcal{P}_\text{small}$ makes up only 20\% of the full dimensionality of the circuit parameter $\bm\theta$. This percentage steadily increases as $p$ moves towards $1$, i.e., reducing the entangling capability of the circuit. It shows that the energy landscape of low-entangling circuits is surprisingly similar to that of over-parameterized systems where most of the variables are used up to parameterize the flat directions. We will observe this geometric resemblance also in the local geometry of intermediate points on the VQA parameter trajectory, serving as the ground for the optimization efficiency of low-entangling circuits. See \cite{liu2021loss} for how efficiently the gradient descent performs the energy optimization in over-parameterized systems.

It is also informative to examine how aligned an initial gradient vector ${\nabla} {\cal L}(\bm\theta)$ is with $\mathcal{P}_\text{small}$. For each different $p$, we estimate the overlap by projecting the normalized gradient onto the subspace $\mathcal{P}_\text{small}$ and computing the norm, written as
\begin{align}
    \frac{1}{\Vert{{\nabla} {\cal L}(\bm\theta)}\Vert} \sqrt{\textstyle\sum_{v \in \mathcal{P}_\text{small}}\big(v\cdot {\nabla} {\cal L}(\bm\theta)\big)^2},
    \label{eq:overlap_psmall}
\end{align}
where the sum is over the orthonormal basis of the space $\mathcal{P}_\text{small}$.
Figure~\ref{fig:hess_ising_g1_init_overlap} exhibits that the lower the entangling capability of the circuit, the smaller the overlap between ${\nabla} {\cal L}(\bm\theta)$ and  $\mathcal{P}_\text{small}$ despite higher dimensionality of $\mathcal{P}_\text{small}$.
It is another characteristic of the low-entangling circuits, 
which is also found in over-parameterized classical deep learning systems \cite{fort2019emergent, gurari2018gradient},
contributing to the fast convergence of the gradient descent minimization of the energy. % ${\cal L}(\bm\theta)$.

\subsection{Optimization Trajectory}
Initial circuit states generated at random parameters, $\bm\theta_0 \sim \mathcal{U}(0,2\pi)^{\otimes nL}$, can typically be well described by the average characteristics of the quantum energy landscape  studied so far. We now turn to investigate the properties of those intermediate circuit states $|\psi(\bm\theta_\tau)\rangle$ obtained after $\tau$ steps of the gradient descent update. 

Unlike the initial circuit energy $\mathcal{L}(\bm\theta_0)$ that cannot differ much from the ensemble average \eqref{eq:bp-energy-mean}, the energy $\mathcal{L}(\bm\theta_\tau)$ at an intermediate time $\tau$ should significantly deviate almost by definition \eqref{gdm} of the steepest descent method.
It indicates how distinctive the intermediate states $|\psi(\bm\theta_\tau)\rangle$ are from initial states, thus requiring independent exploration of their geometric properties.

% So far, we have discussed the average characteristics of the quantum energy landscape. They represent most of the circuit states generated at generic random points $\bm\theta$ in the parameter space. 
% If one computes the Hamiltonian expectation value for a random circuit state $|\psi(\bm\theta_0)\rangle$ where $\bm\theta_0$ is drawn from the uniform distribution, the energy $\mathcal{L}(\bm\theta_0)$ will not deviate much from  \eqref{eq:bp-energy-mean}, 

\begin{figure*}[t]
\centering
\subfloat[Top/bottom eigenvalues]{
    \includegraphics[height=4.1cm]{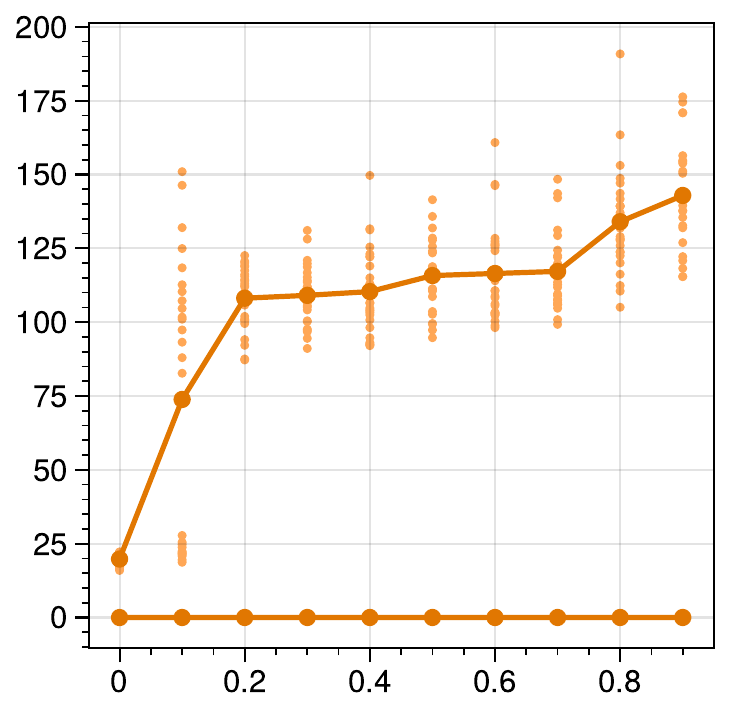}\label{fig:hess_ising_g1_final_topbottom}}
    \subfloat[\% of large eigenvalues ($|\lambda| > 25$)]{
    \includegraphics[height=4.1cm]{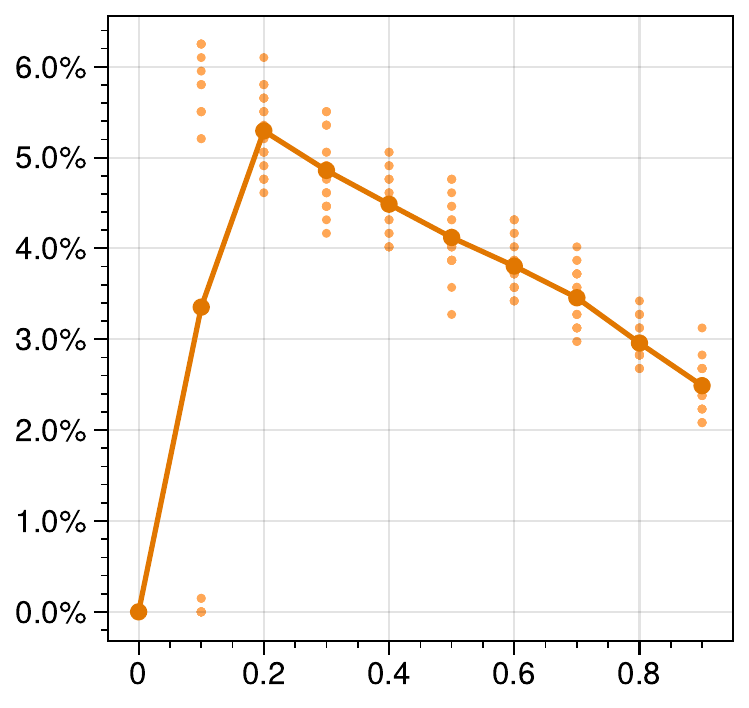}\label{fig:hess_ising_g1_final_large}}
\subfloat[\% of small eigenvalues ($|\lambda| < 0.2$)]{
    \includegraphics[height=4.1cm]{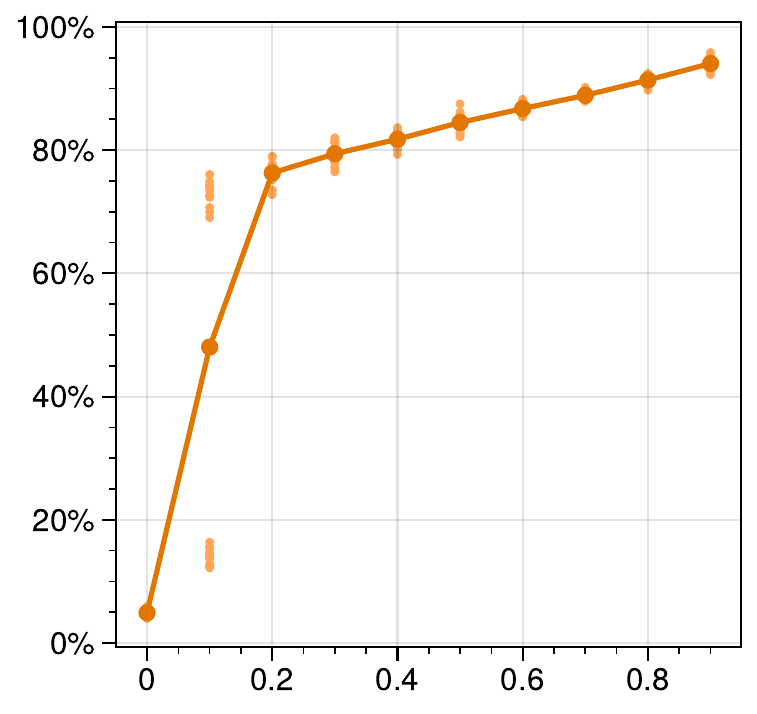}\label{fig:hess_ising_g1_final_small}}
\subfloat[Gradient overlap with $\mathcal{P}_{\text{small}}$]{
    \includegraphics[height=4.1cm]{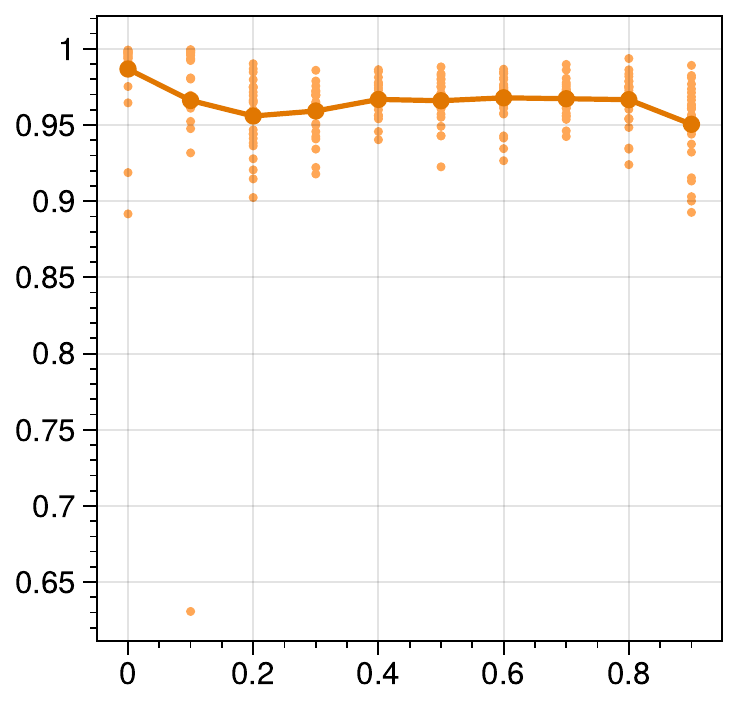}\label{fig:hess_ising_g1_final_overlap}}
    \caption{Characteristic plots for the Hessian eigenspectrum with $n=12$ qubits and $L=56$ layers, based on 50 VQA instances for each probability $p$ to omit the CZ-gates, after 5000 steps of the parameter update.}
    % \caption{Properties of the Hessian eigenvalues spectrum for $n=12$ qubits and $L=56$ layers as a function of the 
    % the probability $p$ for omitting a CZ-gate (after optimization).}
    \label{fig:hess_ising_g1_final}
\end{figure*}

\subsubsection{Optimization Rate}

The gradient descent aims to solve the task of minimizing the energy function by reaching $\bm\theta^* = \arg\min_{\bm\theta}\mathcal{L}(\bm\theta)$ through the iterative and discrete parameter updates,
\begin{align}
\begin{split}
\mathbf{v}_{\tau+1} &= \beta \mathbf{v}_{\tau} + \eta \nabla {\cal L}(\bm\theta_\tau) \ ,\\
\bm\theta_{\tau+1} &= \bm\theta_\tau +  \mathbf{v}_{\tau+1} \ ,
\end{split}.
\label{gdm}
\end{align}
where we set the learning rate $\eta$ and momentum coefficient $\beta$ to be $(\eta, \beta) = (0.9, 0.01)$ throughout all numerical experiments in the paper. Taking the continuum limit, (\ref{gdm}) turns into the gradient flow equation, written as \cite{DBLP:journals/corr/abs-2012-04728}
\begin{equation}
(1-\beta) \frac{d \theta_a}{d \tau}  = - \partial_a  \mathcal{L}(\theta) \ .
\label{gm}
\end{equation}
Any operator in the system shows no explicit dependence on $\tau$. Therefore, the chain rule implies that
\begin{equation}
 (1-\beta)\frac{d}{d \tau} = (1-\beta)\sum_{a=1}^{nL} \frac{d\theta_a}{d\tau}\frac{\partial}{\partial \theta_a} = -(\nabla \mathcal{L}(\bm\theta) \cdot \nabla) \ .
\end{equation}
As a consequence, the average rate of an operator ${\cal O}$ along the optimization trajectory is
\begin{align}
    \mathbb{E}_{\bm\theta}\bigg[\frac{d {\cal O}}{d \tau}\bigg] =  \mathbb{E}_{\bm\theta}\bigg[\frac{\nabla {\cal L}(\bm\theta)\cdot \nabla {\cal O}}{\beta-1} \bigg] =
    \mathbb{E}_{\bm\theta}\bigg[ \frac{ {\cal O}\nabla^2 {\cal L}(\bm\theta)}{1-\beta} \bigg].
%   &=&  -  N E_{\theta}\left({\cal L} {\cal O} \right) \ ,
\end{align}
Note that the last equality is obtained after the integration by parts, where the averaging integral over the compact space $\bm\theta\sim T^{nL}$ cannot produce a boundary term.

% In the following the overall factor $(1-\beta)$ will not matter and for simplicity of notation 
% we will use the continuous gradient flow equation:
% \begin{equation}
% \frac{d \theta_i}{d t}  = - \nabla_i  \mathcal{L}(\theta) \ .
% \label{g}
% \end{equation}
% We thus define:

% where here and in the following (unless stated explicitly) we use the Einstein summation convention.

\begin{theorem}
The following  statements hold along  the optimization trajectory $\{\bm\theta_\tau\}_{\tau}$ parameterized by discrete integer steps $\tau$:
\\

\noindent{\normalfont(i)} The optimization rate satisfies:
\begin{equation}
\Pr\left(\left\vert\frac{d \bm\theta}{d \tau}\right\vert \geq c \right) \leq   \frac{\Var_{\bm\theta}(\nabla {\cal L}(\bm\theta))}{c^2(1-\beta)^2 }
\end{equation}

\noindent{\normalfont(ii)} The average rate of the energy function is:
\begin{equation}
\mathbb{E}_{\bm\theta} \bigg[\frac{d {\cal L}(\bm\theta)}{d \tau}\bigg] =  -\frac{1}{1-\beta}\sum_{a=1}^{nL} \Var_{\bm\theta}(\partial_a {\cal L}(\bm\theta)) \ .
\label{th3}
\end{equation}

\noindent{\normalfont(iii)} The average rate of the energy gradient vanishes:
\begin{equation}
\mathbb{E}_{\bm\theta} \left(\frac{d  \nabla{\cal L}(\bm\theta)}{d \tau}\right) = 0 \ .
\label{th4}
\end{equation}

\noindent{\normalfont(iv)} The average rates of the Hessian and the higher-order energy derivatives $T_{a_1\cdots a_k}(\bm\theta) = \partial_{a_1} \cdots\partial_{a_k}\mathcal{L}(\bm\theta)$  are:
\begin{align}
\mathbb{E}_{\bm\theta} \bigg[\frac{d H_{ab}(\bm\theta)}{d \tau}\bigg] &=    \,\mathbb{E}_{\bm\theta}\bigg[\frac{\sum_{c=1}^{nL}H_{ac}(\bm\theta)H_{cb}(\bm\theta)}{1-\beta}\bigg]  \label{th5}\\
\mathbb{E}_{\bm\theta} \bigg[\frac{d T_{a_1\cdots a_k}(\bm\theta)}{d \tau}\bigg] &=   \mathbb{E}_{\bm\theta}\bigg[\frac{\sum_{c=1}^{nL}T_{a_1\cdots a_{k-1}c}(\bm\theta)H_{ca_k}(\bm\theta)}{1-\beta}\bigg]  \label{th5-high}
\end{align}

% \begin{eqnarray}
% \mathbb{E}_{\bm\theta} \left(\frac{d H_{ij}(\theta)}{d t}\right) &=&   nL \mathbb{E}_{\bm\theta}(\nabla_i {\cal L}\nabla_j {\cal L})  \nonumber\\
% \mathbb{E}_{\bm\theta} \left(\frac{d~ Tr H (\theta)}{d t}\right) &=&   nL \sum_{i=1}^N \Var_{\bm\theta}(\nabla_i {\cal L}(\theta)) \ .
% \label{th5}
% \end{eqnarray}

\end{theorem}

\begin{proof} (i)  The inequality follows from the gradient flow equation (\ref{gm})
inserted into the Chebyshev inequality:
\begin{equation}
    \Pr( |\partial_a  \mathcal{L}(\bm\theta)| \geq c) \leq \frac{\Var_{\bm\theta}(\partial_a {\cal L}(\theta))}{c^2} \ .
\end{equation}

\noindent(ii) Along the optimization curve,
\begin{equation}
  \frac{d {\cal L}(\bm\theta)}{d \tau} =  \frac{\nabla \mathcal{L}(\bm\theta) \cdot \nabla \mathcal{L}(\bm\theta)}{\beta-1} \ .
  \label{changeL}
\end{equation}
By taking the expectation value $\mathbb{E}_{\bm\theta}$ on both sides of (\ref{changeL}), we arrive at (\ref{th3}) thanks to \eqref{eq:bp-mean} that $\mathbb{E}_{\bm\theta}[\partial_a {\cal L}(\bm\theta)] = 0$. 

%\end{proof}

\noindent(iii) Along the optimization curve,
\begin{equation}
  \frac{d \partial_{a}{\cal L}(\bm\theta)}{d \tau} 
  %=  \sum_{b=1}^{nL} \partial_a \partial_{b}{\cal L}(\bm\theta)\,\partial_b{\cal L} (\bm\theta)
  = \frac{1}{2} \sum_{b=1}^{nL} \partial_{a} \left(\partial_b {\cal L} (\bm\theta)\right)^2 \ .
  \label{changeLL}
\end{equation}
The RHS of (\ref{changeLL}) vanishes upon taking the expectation value $\mathbb{E}_{\bm\theta}$, as it is a total derivative on the compact torus. 

\noindent (iv) Along the optimization curve,
\begin{align}
\begin{split}
  \frac{d H_{ab}} {d \tau} &=   \frac{ \nabla H_{ab} \cdot \nabla{\cal L}}{\beta-1} \\
  \frac{d T_{a_1\cdots a_k}} {d \tau} &=   \frac{ \nabla T_{a_1\cdots a_k} \cdot \nabla{\cal L}}{\beta-1} \ .
\end{split}
  \label{changeH}
\end{align}
By taking the expectation value $\mathbb{E}_{\bm\theta}$ on both sides of (\ref{changeH}) and doing the integration by parts, we find (\ref{th5}) and \eqref{th5-high}, for which the integral over the torus $\bm\theta\sim T^{nL}$ does not make a boundary contribution.

% \begin{eqnarray}
%   \mathbb{E}_{\bm\theta} \left(\frac{d H_{ij}}{d t}\right) &=&  \mathbb{E}_{\bm\theta}\left( H_{ij}  \nabla_{k}\nabla_k {\cal L}\right) = 
%   -  nL \mathbb{E}_{\bm\theta} ( H_{ij} {\cal L}) = \nonumber\\ &=&  nL \mathbb{E}_{\bm\theta}(\nabla_i {\cal L} \nabla_j {\cal L})
%   \ ,
%   \label{changeH1}
% \end{eqnarray}
% where we used integartion by parts in $\theta$ as well as the shift rule for to write
% $\nabla_{k}\nabla_k {\cal L} = - {\cal L}$ and a summation over $k$ to get the factor $nL$.
% Setting $i=j$ in (\ref{changeH1}) and summing over the index $i$ we get the second result in (\ref{th5}).

\end{proof}
Alongside the exponential decay \eqref{eq:bp-vargrad} of $\Var_{\bm\theta}[\partial_a \mathcal{L}(\bm\theta)]$ in highly expressible circuits \cite{McClean2018bp} and the consequent suppression \eqref{r2} of the Hessian elements \cite{cerezo2020hessian_bp}, these theorems \eqref{th3}, \eqref{th5}, and \eqref{th5-high} highlight the issue of trainability at initial steps $\tau$ when circuit-generated states $\{|\psi(\bm\theta_\tau)\rangle\}_{\tau}$ are similar to the average Haar-random states.

\begin{figure*}[t]\centering
\subfloat[Top eigenvalue $\lambda_\text{max}$]{
    \includegraphics[height=4.1cm]{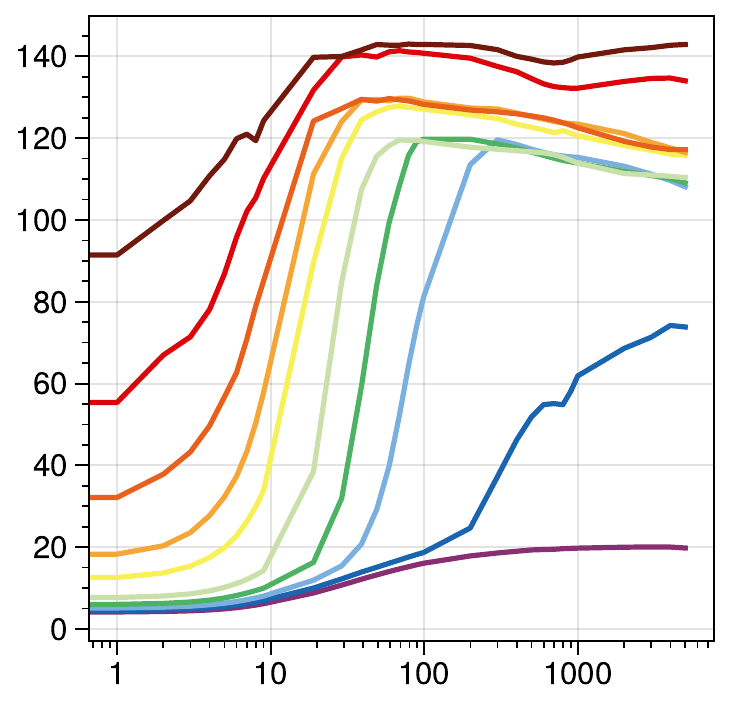}\label{fig:hess_ising_g1_ev_maxeigval}}
\subfloat[Bottom eigenvalue $\lambda_\text{min}$]{
    \includegraphics[height=4.1cm]{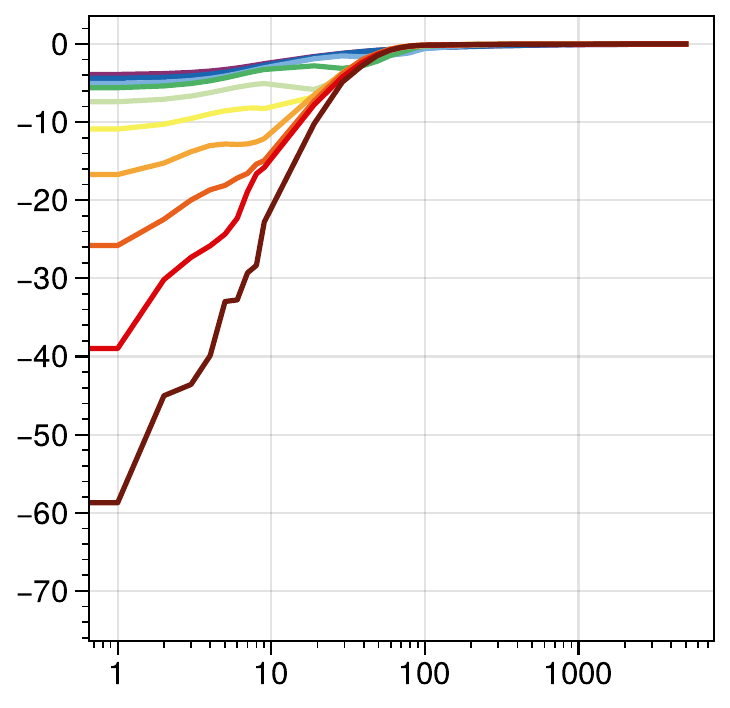}\label{fig:hess_ising_g1_ev_mineigval}}
\subfloat[\% of large eigenvalues ($|\lambda| > 25$)]{
    \includegraphics[height=4.1cm]{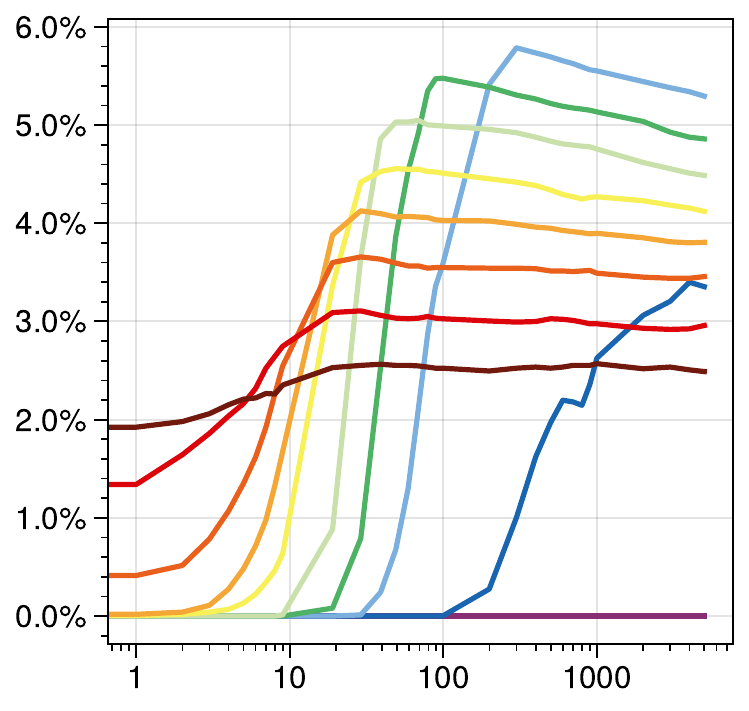}\label{fig:hess_ising_g1_ev_flat_large}}
\subfloat[\% of small eigenvalues ($|\lambda| < 0.2$)]{
    \includegraphics[height=4.1cm]{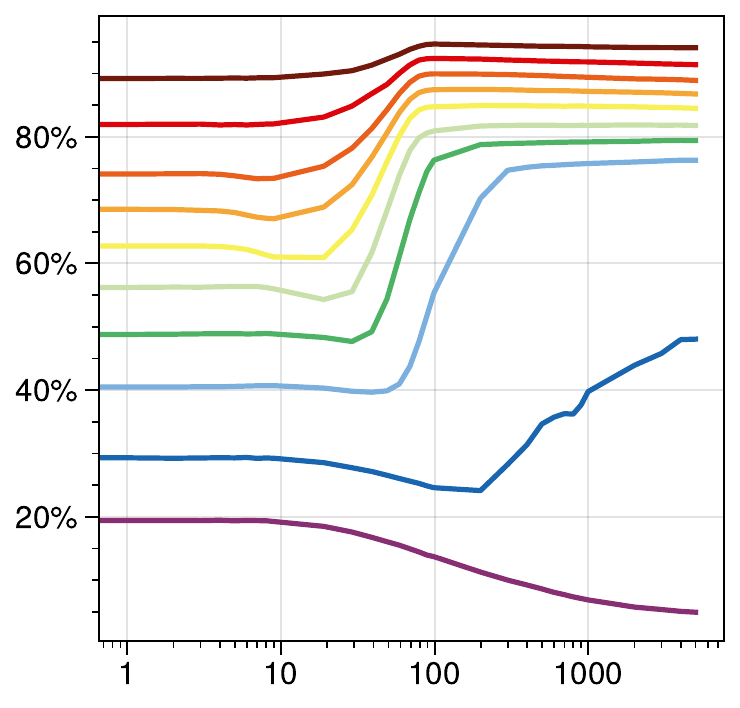}\label{fig:hess_ising_g1_ev_steep_small}}     \\  
\subfloat[Gradient overlap with $\mathcal{P}_{\text{small}}$]{
    \includegraphics[height=4.1cm]{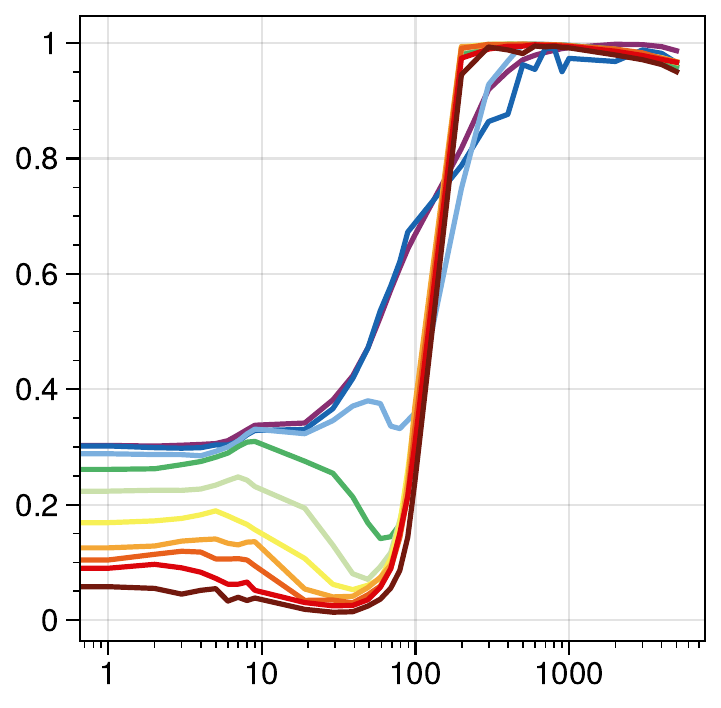}\label{fig:hess_ising_g1_ev_flat_overlap}}
\subfloat[Gradient overlap with $\mathcal{P}_{\text{large}}$]{
    \includegraphics[height=4.1cm]{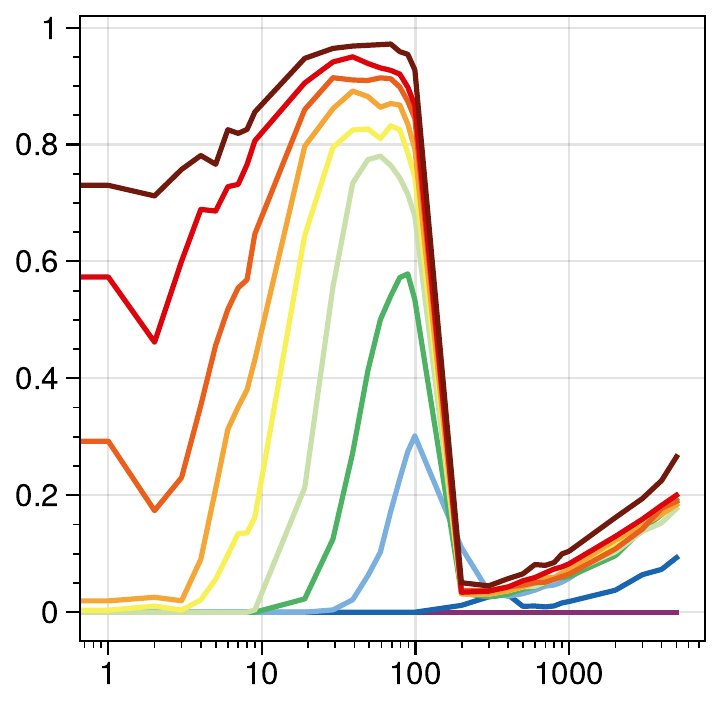}\label{fig:hess_ising_g1_ev_steep_overlap}}       
\subfloat[Energy difference $\Delta E$]{
    \includegraphics[height=4.1cm]{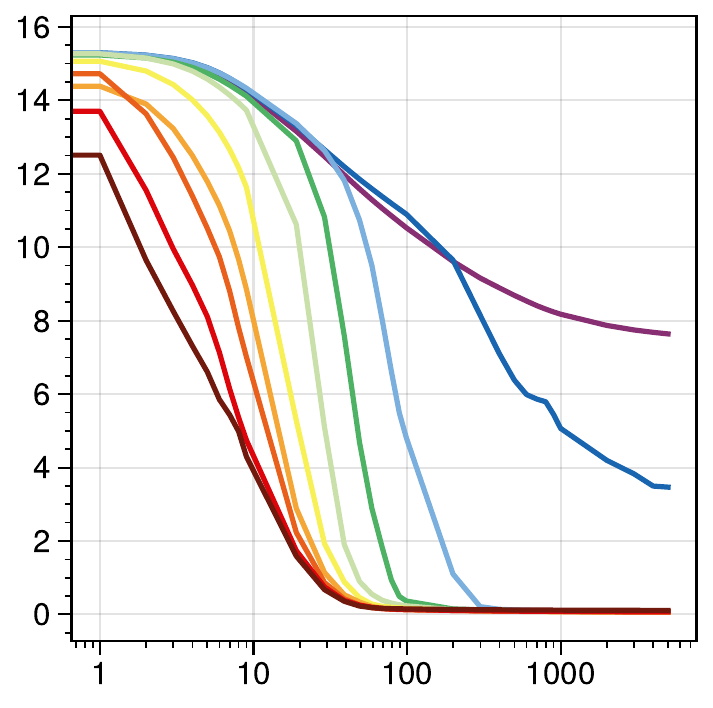}\label{fig:hess_ising_g1_ev_energy}}
\subfloat[Renyi-2 entropy $\mathcal{R}^{(2)}$]{
    \includegraphics[height=4.1cm]{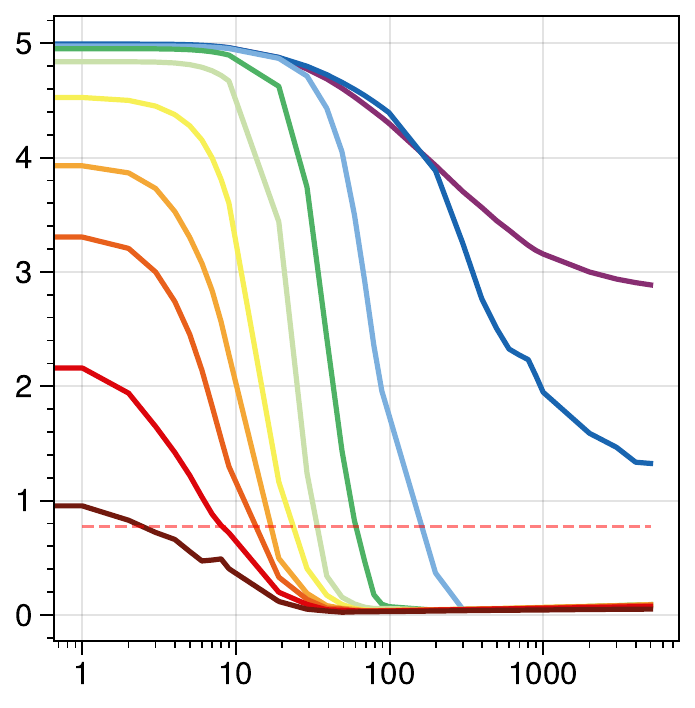}\label{fig:hess_ising_g1_ev_renyi2}}
\quad\ \ \includegraphics[width=17cm]{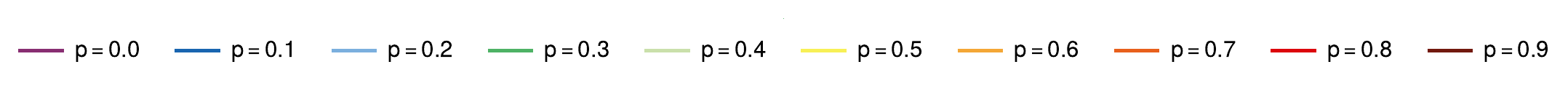}
    \caption{The evolution of the Hessian eigenvalues, the overlap of the energy gradient with the small/large curvature subspaces, the energy gap from the ground state, and the Renyi-2 entropy.
    Their estimation is based on a collection of $26$ Hessian samples for every $0 \leq p < 1$ and $\tau \in \{a\times 10^{b}: a\times 10^{b} \leq 5000 \text{ and } 0 \leq a,b \leq 9\}$.}
    \label{fig:hessian-ev}
\end{figure*}

\subsubsection{Geometry at Endpoints}
We now characterize the local geometry near the optimization endpoint $\bm\theta_{5000}$ by collecting $260$ sample Hessian eigenspectra evaluated after $5000$ steps of the parameter update \eqref{gdm}. They are visualized in Figure~\ref{fig:hess_ising_g1_final_spectrum} at a glance, whose horizontal and vertical directions extend across the spectral values and circuit instances, being distinguished by $10$ different colors according to the $p$ values. Besides, their top/bottom eigenvalues are also plotted in Figure~\ref{fig:hess_ising_g1_final_topbottom} as a function of $p$. Having inspected these spectral data, we make the following observations:

First, all negative eigenvalues in the Hessian spectrum are, if not zero, negligible in their absolute values. This illustrates that the local geometry around $\bm\theta_{5000}$ no longer contains a concave direction because the gradient descent has already converged to a local extremum. 

Second, the Hessian spectrum at $\bm\theta_{5000}$ distributes more widely as the circuit entangling capability decreases, i.e., $p \rightarrow 1$. All the top eigenvalues at $p=0$ are upper bound by $25$, while the top eigenvalues for $p \geq 0.1$ are frequently greater than $100$. Such widespread/concentration of the Hessian spectrum correlated to the entangling capability was also observable at the initial points $\bm\theta_{0}$.

Third, the top Hessian eigenvalues across all $0 \leq p < 1$ can be roughly classified into two clusters, i.e., $h_\text{top} \sim 20$ and  $h_\text{top} \gtrsim 100$, connected to the success/failure of the VQA optimization in approximating the ground state. 
It is a notable distinction from the initial Hessian spectrum where the top eigenvalues gradually increase from $h_\text{top} \sim 5$ to $h_\text{top} \gtrsim 70$ as $p$ approaches to $1$.

One can also infer the geometric structure of the trajectory endpoint $\bm\theta_{5000}$ by examining the percentages of small and large eigenvalues in the Hessian eigenspectrum. We regard a Hessian eigenvalue $h_a$ as small if $|h_a| < 0.2$ and large if $|h_a| > 25$, such that no eigenvalue at $p=0$ can be classified as large. 
The corresponding fractions of large/small eigenvalues are depicted in Figures~\ref{fig:hess_ising_g1_final_large}~and~\ref{fig:hess_ising_g1_final_small}. 
We again find two clusters therein according to the success/failure of VQA samples.

Some circuit instances with high entangling capability, i.e., $p\leq 0.1$, converge to non-optimal extrema, whose local Hessian spectrum contains no large and only a few small values. 
Hence, the shape of non-optimal extrema is mildly convex and nearly isolated, as there are no steep directions and only a handful of flat directions.

In contrast, the local geometry at the endpoints, $\bm\theta_{5000}$, close to the ground-level energy $E(\bm\theta_{5000}) \simeq E_g$ exhibits the dominance of the flat directions, increasing from 70\% to 95\% of the total dimension of circuit parameters as $p$ grows. 
%a large proportion of small Hessian eigenvalues, 
% Such  in the Hessian spectrum implies that 
Then the subspace spanned by the large or intermediate eigenvectors makes only up a smaller portion of the parameter space dimensions. 
% It suggests that the Hessian spectrum of the circuits with low-entangling capability is composed of bulk eigenvalues around $0$ and a few outlier eigenvalues far away from $0$.
Accordingly, the optimal minima should
%$\bm\theta_{5000}$ appears to be 
resemble a steep-sided valley with exceedingly high-dimensional flat directions.

\begin{figure*}[t]
\centering
\subfloat[Energy difference $\Delta E$]{
    \includegraphics[height=4.1cm]{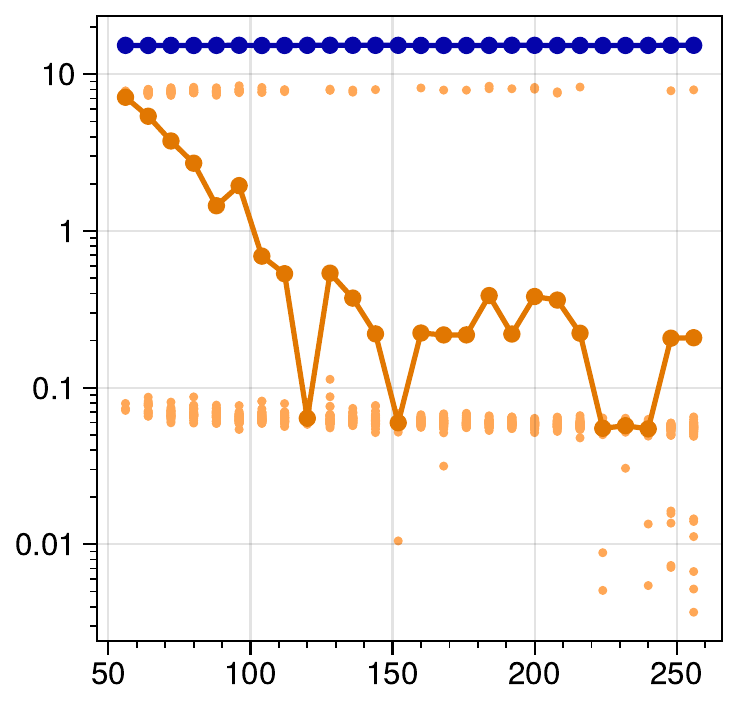}\label{fig:optim_ising_g1_left_param}}
\subfloat[Renyi-2 entropy $\mathcal{R}^{(2)}$]{
    \includegraphics[height=4.1cm]{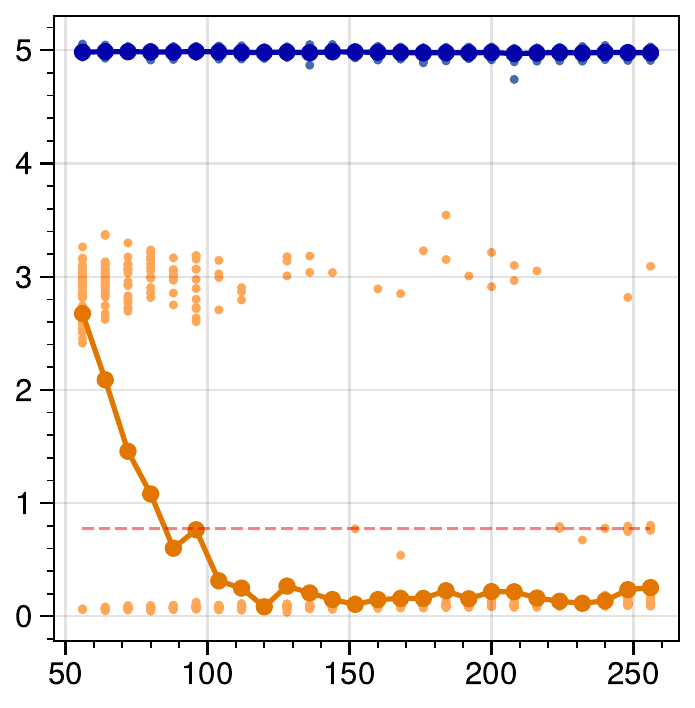}\label{fig:optim_ising_g1_middle_param}}
\subfloat[\% that reaches $\Delta E < 0.1$]{
    \includegraphics[height=4.1cm]{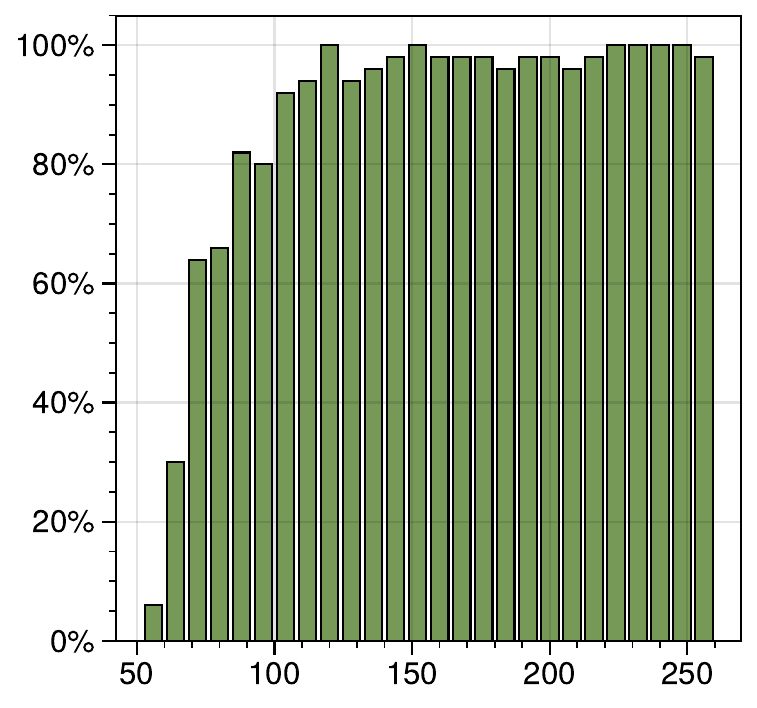}\label{fig:optim_ising_g1_success_param}}
\subfloat[Steps $\tau$ until $\Delta E < 0.1$]{
    \includegraphics[height=4.1cm]{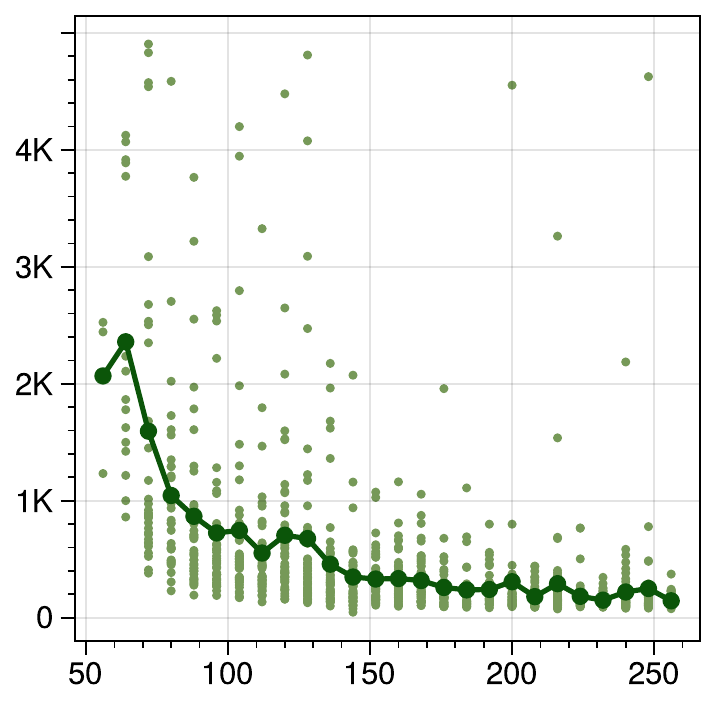}\label{fig:optim_ising_g1_right_param}}
    \caption{Collections of 50 VQA instances to approximate the $n=12$ Ising ground state for each circuit depth $L$ that counts both $56$ entangling and $(L-56)$ rotation layers. Adding control parameters with a fixed amount of the entanglement capability improves the optimization performance. }
    % \caption{The VQA result with $n=12$ qubits, $L=56$ entangling layers
    % and additional rotation layers as a function of the total number of layers.  We see that adding control
    % parameters with a fixed entanglement structure improves the optimization performance. (a) The energy difference between the optimized circuit and the ground state.
    % (b) The Renyi-2 entropy of the optimized circuit. (c) The percentage of successful optimizations. (d) The required number of optimization steps.}
    \label{fig:optim_ising_g1_param}
\end{figure*}

We notice the similarity in the geometric structure near the optimization endpoints between the quantum energy landscape with low-entangling capability, i.e., $p\rightarrow 1$, and classical over-parameterized systems such as deep neural networks \cite{sagun2017eigenvalues}. It is tempting to speculate that the variational circuit with lower  entanglement capability can effectively enter the over-parameterized regime with fewer parameters, which explains why the local gradient search can quickly reach the optimal parameter $\bm\theta^*$ that corresponds to the Hamiltonian ground state  $|\psi(\bm\theta^*)\rangle \simeq |\Psi\rangle$ \cite{liu2021loss, highdepth}. 
Somewhat tangentially, the effect of adding more circuit parameters while keeping the same amount of the entangling capability will be considered in Section~\ref{sec:param}.

\subsubsection{Evolution of the Geometry}
We extend the investigation on the Hessian spectrum to various intermediate points along the optimization trajectory $\{\bm\theta_\tau\}_{\tau=1}^{5000}$ at different steps $\tau$.
Figure~\ref{fig:hessian-ev} illustrates the evolution of the local geometry shown through a sample collection of $26$ Hessian eigenspectra for each $0 \leq p < 1$ at $\tau \in \{a\times 10^{b}: a\times 10^{b} \leq 5000 \text{ and } 0 \leq a,b \leq 9\}$.
Besides reconfirming how the entangling capability affects the local energy landscape around initial/final points, we also make the following new observations from that:

First, Figures~\ref{fig:hess_ising_g1_ev_maxeigval}~and~\ref{fig:hess_ising_g1_ev_mineigval} exhibits the abrupt rise of top eigenvalues followed by the slow adjustment and also the regular convergence of bottom eigenvalues to $0$, especially for those successful VQA instances with $\mathcal{L}(\bm\theta_{5000}) \simeq E_g$. 
It corresponds to the rapid movement of $\bm\theta_\tau$ rolling down into an attractor basin and then fine-tuning itself to minimize the energy $\mathcal{L}(\bm\theta_\tau)$ inside the basin.

Second, Figures~\ref{fig:hess_ising_g1_ev_maxeigval} shows that the surge of top eigenvalues occurs sooner with the lower entangling capability. For example, the average number of steps $\tau$ until the rise of $h_{\text top}$ reduces from a few hundreds down to a few tens as $p$ grows from $0.2$ to $0.7$. 
We remark that the rapid convergence of the gradient descent is another resemblance between the low-entangling circuits and over-parametrized systems \cite{liu2021loss}.

Third, Figures~\ref{fig:hess_ising_g1_ev_flat_large}~and~\ref{fig:hess_ising_g1_ev_steep_small} show the gradual crossover of the landscape geometry into either a steep canyon or a convex bowl along the optimization trajectory, depending on whether the energy at the convergence point $\mathcal{L}(\bm\theta_{5000})$ is sufficiently close to or far from the ground energy $E_g$.

Fourth, Figures~\ref{fig:hess_ising_g1_ev_flat_overlap}~and~\ref{fig:hess_ising_g1_ev_steep_overlap} display a sharp transition in the alignment of the gradient vector $\nabla\mathcal{L}(\bm\theta_\tau)$ through \begin{align}
    \frac{1}{\Vert{{\nabla} {\cal L}(\theta)}\Vert} \sqrt{\textstyle\sum_{v \in \mathcal{P}_\text{small/large}}\big(v\cdot {\nabla} {\cal L}(\theta)\big)^2},
\end{align}
which is the norm of the unit gradient projected onto the subspaces $\mathcal{P}_\text{small/large}$ spanned by the small/large Hessian eigenvectors. The parameter update only makes  $\nabla\mathcal{L}(\bm\theta_\tau)$ more aligned with $\mathcal{P}_\text{large}$ during the initial stage. Afterwards, the gradient overlap with $\mathcal{P}_\text{small/large}$ surges/drops  suddenly around the transition point $\tau_t \simeq 100$, indicating the initial phase of the gradient descent is over, and the parameter $\bm\theta_{\tau \geq \tau_t}$ has been already confined to an attractor basin. We comment that such existence of two phases means the gradient descent is not always aligned with the subspace $\mathcal{P}_\text{large}$, not as for the deep neural networks \cite{gurari2018gradient}.

\section{Control Parameters}
\label{sec:param}
This section studies the effect of the number of control parameters on the quantum energy landscape and VQA performance. To run controlled experiments with a fixed degree of the entangling capability, 
we will always start with the $L^*=56$ circuits in Figure~\ref{fig:circuit_diagram} and introduce extra parameters by randomly sandwiching between the $L^*$ layers $n$ copies of the Pauli-$y$ rotation gate acting on every qubit  $(L - L^*)$ times. The total number of circuit parameters is therefore $nL = nL^* + n(L - L^*)$. Notice that the additional parameters are redundant as Pauli-$y$ rotations can commute, i.e., 
\begin{align}
\begin{split}
    R_{y,i}(\varphi_1) R_{y,i}(\varphi_2) &=R_{y,i}(\varphi_2) R_{y,i}(\varphi_1) \\&= R_{y,i}(\varphi_1 +\varphi_2)
\end{split}
\label{eq:red}
\end{align}

To see if the enlarged parameter space significantly impacts the VQA performance, we start with $50$ random circuit instances and minimize the mean energy $\mathcal{L}(\bm\theta)$ of the Ising Hamiltonian \eqref{eq:ising} by the gradient descent method. All the VQA optimization results for $n=12$ qubits are summarized in Figure~\ref{fig:optim_ising_g1_param} as a function of the total number of layers $L$. We check that more control parameters, despite their redundancy \eqref{eq:red}, can still facilitate the local gradient search of the optimal parameters: A cluster of orange dots near $\Delta E \sim 9$, which represents unsuccessful attempts in reaching the ground-level energy $E_g$, becomes less populated as $L$ increases. Getting deeper not only increases the rate of optimization success, defined by $\Delta E < 0.1$, but also reduces the parameter update steps $\tau$ required for it. Moreover, with exceedingly many parameters, such as $L \gtrsim 200$, the variational circuit can often achieve a high precision approximation that even satisfies $\Delta E \lesssim 10^{-2}$ and $\Delta \mathcal{R}^{(2)} \sim 0$ \cite{highdepth}.

Having found the redundant parameters can positively impact the VQA optimization performance, we look into the geometric changes of the quantum energy landscape driven by increasing the parameter space dimension. Let us compute sample collections of $26$ Hessian eigenspectra for each $L \in\{56, 64, \cdots, 112, 120, 144, 168\}$ before/after $5000$ steps of the gradient descent update. Figure~\ref{fig:hess_ising_g1_spectrum_param} show all the Hessian spectra at a glance, whose horizontal and vertical axes extend over the values and circuit instances, colored differently according to their $L$ values. Some supplementary characteristic curves on the sample Hessian spectra are also presented in Figures~\ref{fig:hess_ising_g1_init_param}~and~\ref{fig:hess_ising_g1_final_param}, where the upper bounds $5$ and $25$ of absolute eigenvalues at $L=56$ are respectively referred to distinguish large eigenvalues before and after the optimization.

The main observation that follows from the data is that the overall patterns of Figures~\ref{fig:hess_ising_g1_spectrum}--\ref{fig:hess_ising_g1_final} and Figures~\ref{fig:hess_ising_g1_spectrum_param}--\ref{fig:hess_ising_g1_final_param} are analogous, demonstrating the qualitative similarity between reducing entanglement capability and adding control parameters. 
% The message of this is that the region of low-entangling states in the Hilbert space
% can be effectively represented either by a non overparametrized low-entangled circuit state 
% or by an overparametrized highly-entangled circuit state.
The low-entangling circuits behave under the gradient descent as if they are over-parameterized since they represent only a limited subset in the Hilbert space.
Hence, the implied principle for designing the variational circuit is to avoid both the high entangling capability and over-parameterization if the VQA task involves searching for low-entangled target states. One should rely on the highly-expressible over-parameterized quantum circuits otherwise.
% The implied design principle for VQA is that we should use a  non overparametrized low-entangled circuit state 
% if the VQA task involves a search for a low-entangled state and an overparametrized highly-entangled circuit state otherwise.

%It suggests that the lower entangling circuits can be effectively over-parameterized with fewer variables as %they represent a more limited region of the Hilbert space. 

\begin{figure}[t]
    \centering
    \subfloat[$\tau=0$]{
    \includegraphics[height=4.1cm]{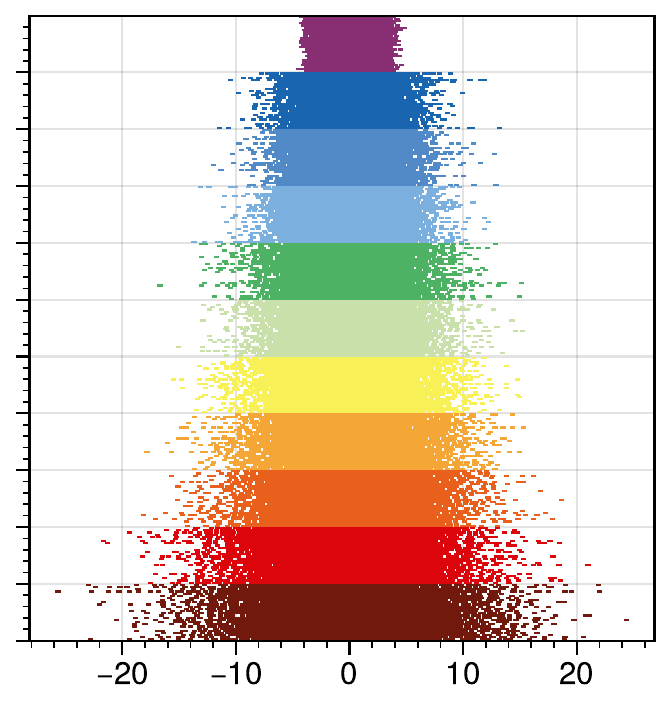}\label{fig:hess_ising_g1_init_spectrum_param}}
    \subfloat[$\tau=5\times 10^3$]{
    \includegraphics[height=4.1cm]{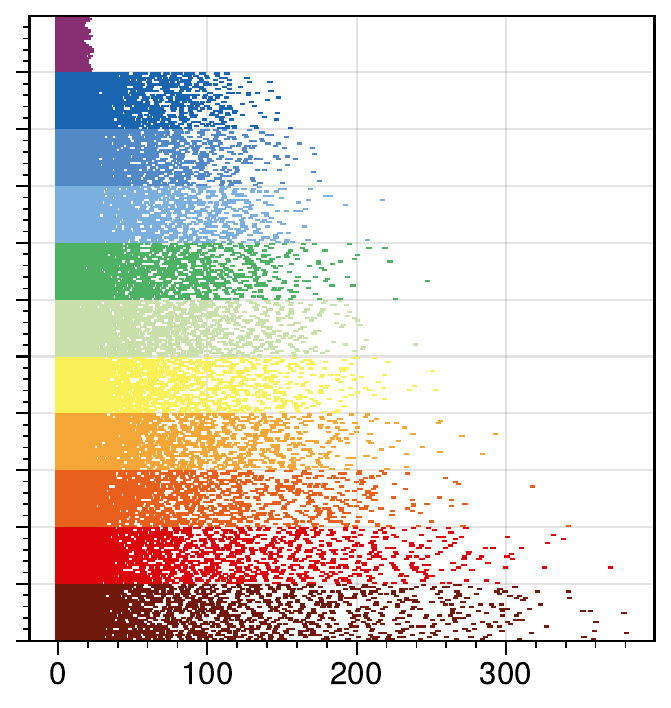}\label{fig:hess_ising_g1_final_spectrum_param}}\\
    \includegraphics[width=\textwidth]{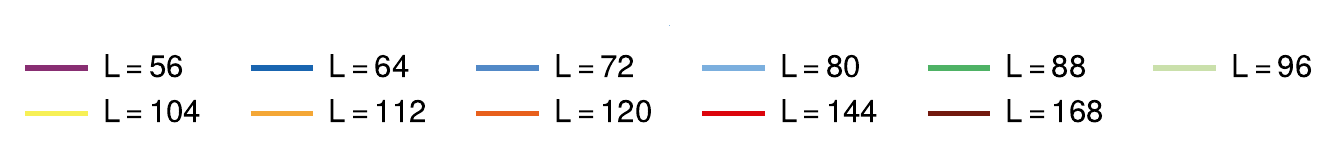}
    \caption{A visualization of 26 sample Hessian eigenspectra, after taking the gradient descent update $\tau$ times, for each circuit depth $L$ that counts both $56$ entangling and $(L-56)$ rotation layers. It looks qualitatively similar to Figure~\ref{fig:hess_ising_g1_spectrum}.}
    \label{fig:hess_ising_g1_spectrum_param}
\end{figure}

\section{Discussion}
\label{sec:discussion}

Throughout this work, we addressed the design principles for quantum variational circuits by proving some theorems and conducting systematic experiments. We demonstrated how the circuit entanglement and the parameter space dimension affect the local geometry of the energy landscape and thus the VQA optimization performance. Our central object of study was the Hessian of the energy function in the parameter space, which shows the curvature eigenspectrum, evaluated at a random initial point and along the optimization trajectory. 

Several analyses on the landscape geometry illustrated that the efficiency of the low-entangling circuits under the VQA optimization is related to their resemblance with an over-parameterized system, despite having a relatively small number of variables. Fewer parameters than $2^n$ may be sufficient to parameterize a subset of Hilbert space that the circuit with limited expressibility can represent. This leaves us with the following question: How many circuit parameters are optimal for a given amount of entangling capability under the gradient descent.\footnote{An analysis of \cite{Funcke_2021} that identifies the parameter redundancy can be useful. See also \cite{Sim_2019,hubregtsen2020evaluation,Rasmussen_2020} for expressibility-related discussions.}

% Hamiltonian 
Another important question unexplored in this paper is the effect of the Hamiltonian on the  energy landscape and optimization accuracy. We recall that the low-entangling circuits are not successful in simulating volume-law entangled ground states \cite{Kim:2021ffs}, whose  entanglement scaling is actually determined by the Hamiltonian. Generally, we would like to characterize how certain defining properties of the Hamiltonian, e.g., locality or degree of spin interactions, can steepen/flatten the energy landscape and thus influence the VQA performance.

Finally, we would like to understand how various types of noise can change the quantum energy landscape \cite{barron2020measurement,zeng2021simulating,fontana2020optimizing}, which may lead to suitable error mitigation schemes for noisy VQA optimization.

\begin{figure*}[t]
\centering
\subfloat[Top/bottom eigenvalues]{
    \includegraphics[height=4.1cm]{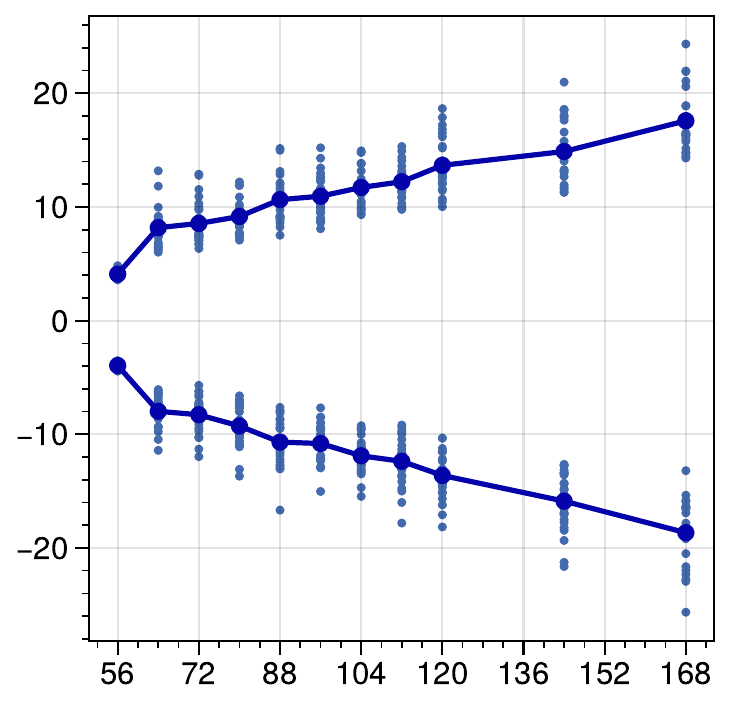}\label{fig:hess_ising_g1_init_outlier_param}}
\subfloat[\% of large eigenvalues ($|\lambda| > 5$)]{
    \includegraphics[height=4.1cm]{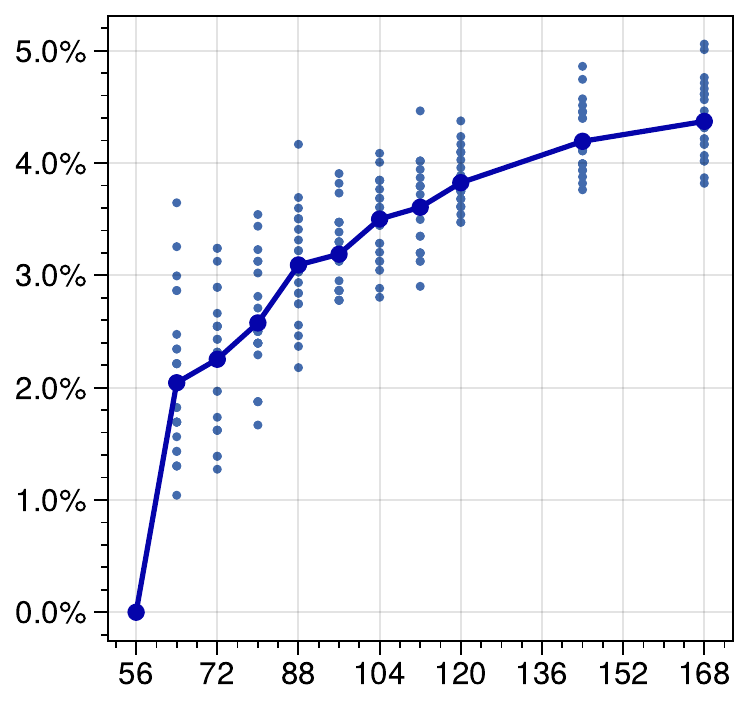}\label{fig:hess_ising_g1_init_large_param}}
\subfloat[\% of small eigenvalues ($|\lambda| < 0.2$)]{
    \includegraphics[height=4.1cm]{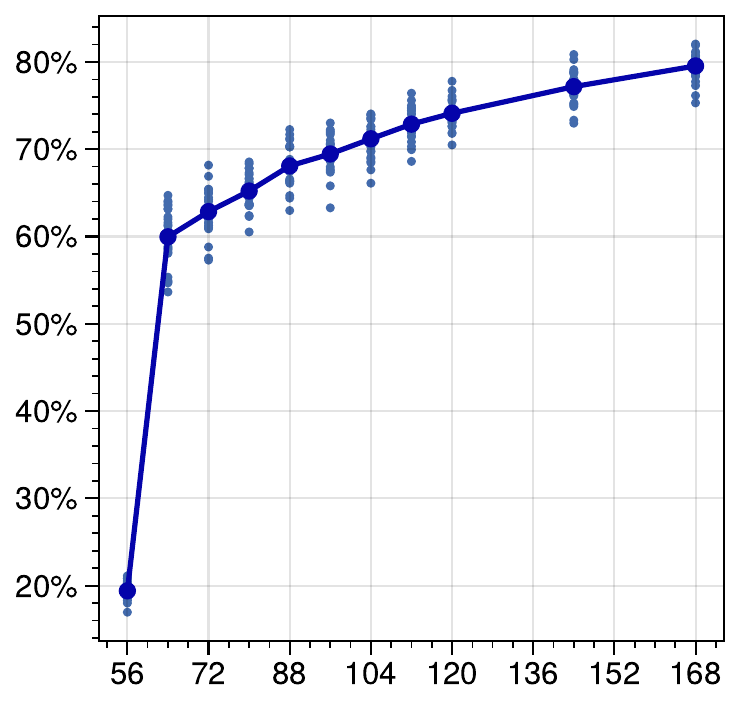}\label{fig:hess_ising_g1_init_small_param}}
\subfloat[Gradient overlap with $\mathcal{P}_{\text{small}}$]{
    \includegraphics[height=4.1cm]{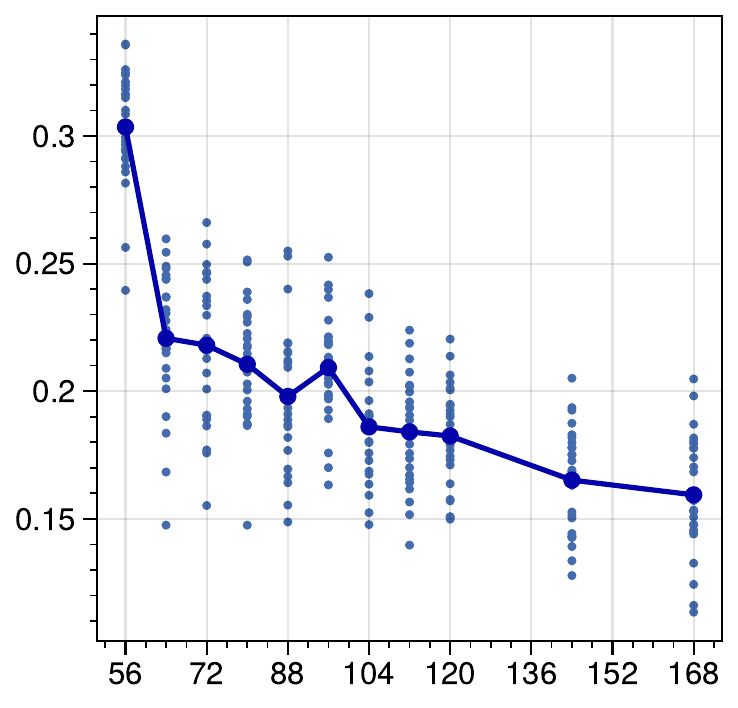}\label{fig:hess_ising_g1_init_overlap_param}}
    \caption{
    % The Hessian eigenspectrum with $n=12$ qubits, $L=56$ entangling layers
    % and additional rotation layers as a function of the total number of layers (before optimization). 
    Characteristic plots for the Hessian eigenspectrum with $56$ entangling and $(L-56)$ rotation layers, based on 26 instances for each $L \in\{56, 64, \cdots, 112, 120, 144, 168\}$, at randomly initialized circuit parameters. 
    Adding control parameters with a fixed amount of the entanglement capability has qualitatively the same effect as fixing the number of control parameters and reducing the entanglement capability. See Figure~\ref{fig:hess_ising_g1_init} that looks qualitatively analogous.}
    \label{fig:hess_ising_g1_init_param}
\end{figure*}

\begin{figure*}[t]
\centering
\subfloat[Top/bottom eigenvalues]{
    \includegraphics[height=4.1cm]{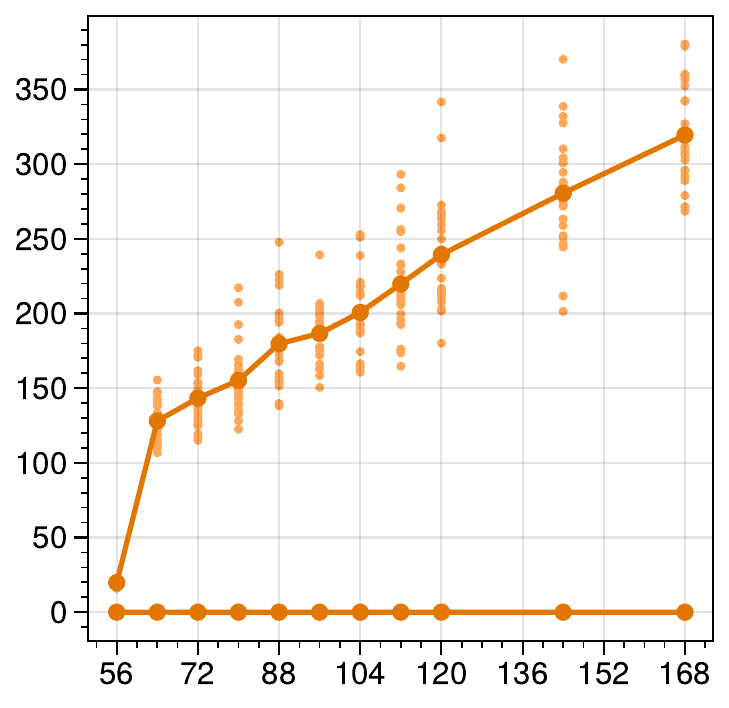}\label{fig:hess_ising_g1_final_topbottom_param}}
    \subfloat[\% of large eigenvalues ($|\lambda| > 25$)]{
    \includegraphics[height=4.1cm]{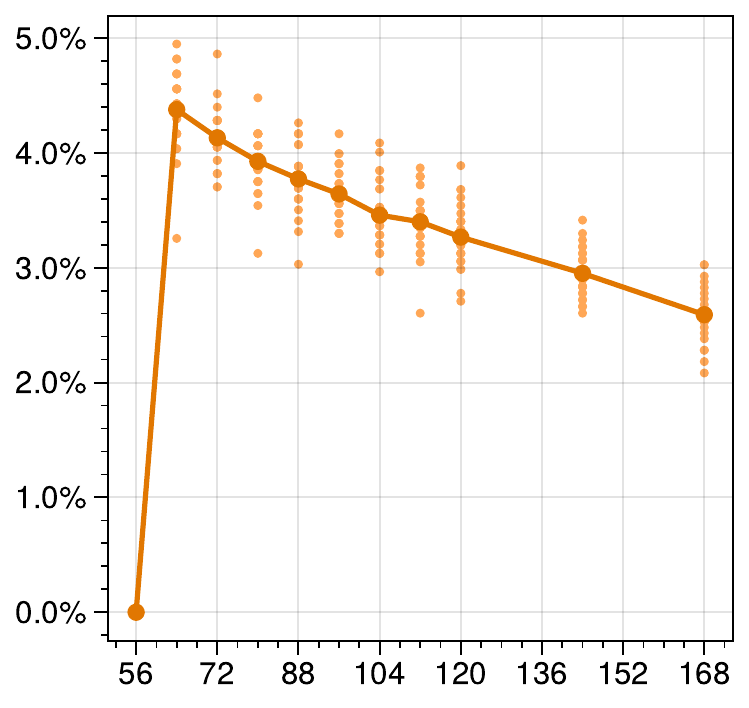}\label{fig:hess_ising_g1_final_large_param}}
\subfloat[\% of small eigenvalues ($|\lambda| < 0.2$)]{
    \includegraphics[height=4.1cm]{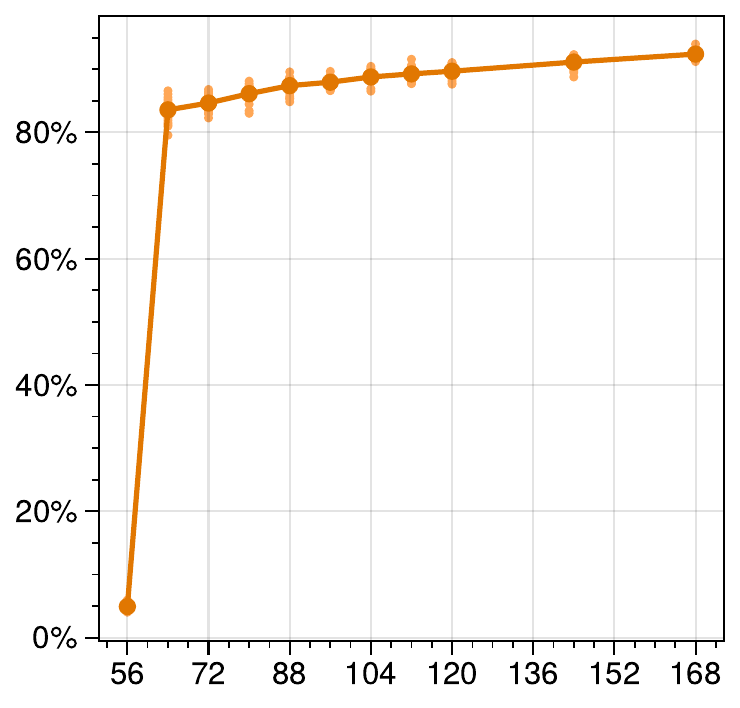}\label{fig:hess_ising_g1_final_small_param}}
\subfloat[Gradient overlap with $\mathcal{P}_{\text{small}}$]{
    \includegraphics[height=4.1cm]{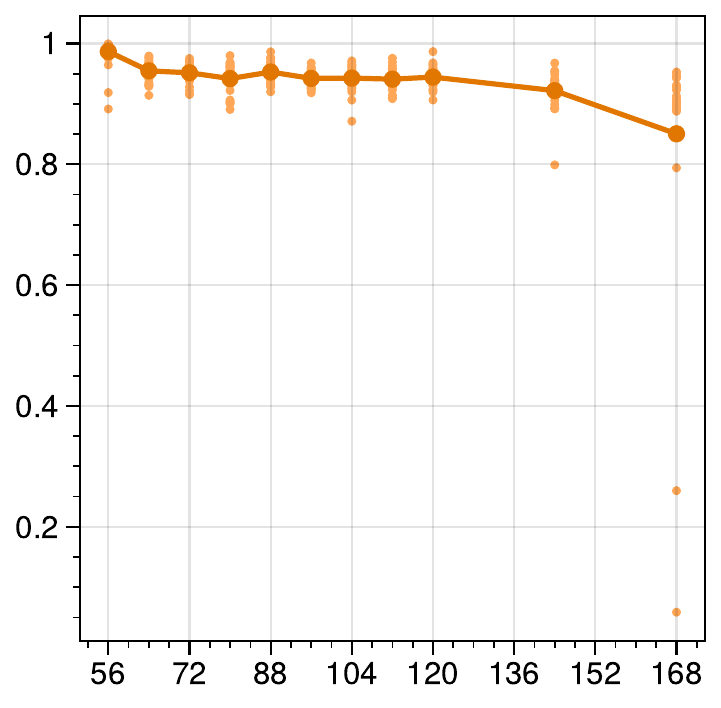}\label{fig:hess_ising_g1_final_overlap_param}}
    \caption{Characteristic plots for the Hessian eigenspectrum with $56$ entangling and $(L-56)$ rotation layers, based on 26 VQA instances for each $L \in\{56, 64, \cdots, 112, 120, 144, 168\}$, after $5000$ steps of the parameter update.
    % \caption{Properties of the Hessian eigenvalues spectrum for $n=12$ qubits,
    % $L=56$ entangling layers
    % and additional rotation layers as a function of the total number of layers (after optimization). 
    They are qualitatively similar to Figure~\ref{fig:hess_ising_g1_final}.}
    \label{fig:hess_ising_g1_final_param}
\end{figure*}

\vfill
\section*{Acknowledgements}
We thank Kishor Bharti, Thi Ha Kyaw, Dario Rosa for valuable discussions.
The work of J.K. is supported by the NSF grant PHY-1911298 and the Sivian fund.
The work of Y.O. is supported in part by  Israel Science Foundation Center
of Excellence, the IBM Einstein Fellowship and John and Maureen Hendricks Charitable Foundation
at the Institute for Advanced Study in Princeton.
Our Python code for the numerical experiments is written in TensorFlow Quantum \cite{TFQ}. The experimental data are managed by using Comet.ML \cite{CometML}.

\pagebreak
\providecommand{\href}[2]{#2}\begingroup\raggedright\endgroup

\end{document}